\documentclass{article}
\usepackage{latexsym,amsmath, amssymb,ulem}
\newcounter{draft} 
\ifnum \thedraft=0
\fi
\makeatletter
\def\subsection{\@startsection {subsection}{2}{\z@}{.8ex
plus 1ex} {1ex}{\sc}}
\makeatother
\newenvironment{proof}{\vspace*{2ex}\noindent {\em Proof:}}{\hfill $\Box$ \\[1ex]}

\newenvironment{acknowledgements}{\\[2ex]\noindent {\it
    Acknowledgements:}}{\\[1ex]}
\newcommand{\new}{\newcommand}
\newcounter{letter}
\newenvironment{alist}{
\begin{list}{{\upshape(\alph{letter})}}{\usecounter{letter}}
}{\end{list}}
\new{\Nabla}{\bigtriangledown}
\new{\Ricci}{\mathrm{Ricci}}
\new{\Pfaff}{\mathrm{Pfaff}}
\new{\Hom}{\mathrm{Hom}}
\new{\MQ}{\mathrm{K}^{\mathrm{qm}}}
\new{\Path}{\mathrm{Path}}
\new{\kernel}{\mathrm{Kernel}}
\new{\boxtensor}{\boxtimes}
\new{\barpsi}{\overline{\psi}}
\new{\barrho}{\overline{\rho}}
\new{\iso}{\cong}
\new{\vol}{\operatorname{Vol}}
\new{\End}{\mathrm{End}}
\new{\grad}{\mathrm{grad}}
\new{\diverge}{\mathrm{div}}
\new{\tr}{\mathrm{tr}}
\new{\indx}{\mathrm{ind}}
\new{\str}{\mathrm{str}}
\new{\Str}{\mathrm{Str}}
\new{\scalar}{{\mathfrak{r}}}
\new{\kk}{{\mathfrak{k}}}
\new{\lop}{{\mathfrak{l}}}
\new{\defequals}{\stackrel{\mathrm{def}}{=}}
\new{\NN}{\mathbb{N}}
\new{\RR}{\mathbb{R}}
\new{\ZZ}{\mathbb{Z}}

\renewcommand{\vector}{\mathbf}
\renewcommand{\vec}{\vector}
\new{\ltilde}{\tilde{L}}
\new{\gap}{{\;}}
\new{\depends}[1]{{\scriptscriptstyle{#1\!\!}}}
\new{\triplegap}{{\;\;\;\;\;}}
\new{\doublegap}{{\;\;\;}}
\new{\FF}{{\mathcal{F}}}
\new{\Alg}{{\mathcal{A}}}
\new{\bundle}{{\mathcal{E}}}
\new{\QQ}{{\mathcal{Q}}}
\new{\CC}{{\mathcal{C}}}
\new{\OO}{{\mathcal{O}}}
\new{\EE}{{\mathcal{E}}}
\new{\D}{{\mathcal{D}}}
\new{\Lie}{{\mathcal{L}}}
\new{\dd}{{\mathfrak{d}}}
\new{\tensor}{\otimes}
\new{\bracket}[1]{\left\langle #1 \right \rangle}
\new{\parens}[1]{\!\left( #1 \right)}
\new{\sbrace}[1]{\!\left[ #1 \right]}
\new{\cbrace}[1]{\!\left\{ #1 \right\}}
\new{\abs}[1]{\left| #1 \r|}
\new{\pt}{\mathfrak{P}}
\new{\ptxystar}{\pt^{x*}_{y}}
\new{\ptyxstar}{\pt^{y*}_{x}}
\new{\bint}{\oint}
\new{\pfend}{\noindent \hfill $\Box$}
\new{\beas}{\begin{eqnarray*}}
\new{\eeas}{\end{eqnarray*}}
\renewcommand{\r}{\right}
\renewcommand{\l}{\left}
\new{\norm}[1]{\l|\!\l| #1 \r|\!\r|}
\new{\para}[2]{\subsection{#1. }
\label{par:#2} \ifnum \thedraft=1 \marginpar{\scriptsize{par:#2}} \fi}
\new{\be}{\begin{equation}  
}
\new{\ee}[1]{
\label{eq:#1}
\end{equation} \ifnum \thedraft=1 \marginpar{\scriptsize{\em{eq:#1}}}  
\fi \noindent%
}
\new{\eqa}[2]{\begin{align}
#2 \label{eq:#1}
\end{align}
\ifnum \thedraft=1 \marginpar{\scriptsize{\em{eq:#1}}}  \fi
}
\new{\eqalabel}[1]{
\label{eq:#1}
}
\new{\dlabel}[1]{\ifnum \thedraft=1 \marginpar{\scriptsize{\em{#1}}}  \fi}
\new{\labitem}[1]{%
\label{item:#1} \ifnum \thedraft=1 \marginpar{\scriptsize{\em{item:#1}}}  
\fi \noindent%
}
\numberwithin{equation}{section}
\newtheorem{remark}{Remark}[section]
\newtheorem{theorem}{Theorem}[section]
\newtheorem{proposition}{Proposition}[section]
\newtheorem{lemma}{Lemma}[section]
\newtheorem{corollary}{Corollary}[section]
\newtheorem{definition}{Definition}
\new{\phone}[1]{Phone: #1}
\new{\fax}[1]{Fax: #1}

\begin{document}

\setcounter{draft}{0} 

\nocite{Atiyah85} 
\nocite{DH82} 
\nocite{BP08}
\nocite{Getzler91}
\nocite{Rogers03}

\ifnum \thedraft=1 \today  \fi

\author{Dana S. Fine\footnote{University of Massachusetts Dartmouth, N. Dartmouth, MA
  02747} and Stephen F. Sawin\footnote{Fairfield University, Fairfield, CT 06824}}
\title
{A Rigorous Path Integral for $N=1$ Supersymmetic Quantum
Mechanics  on a Riemannian Manifold}



\maketitle
\begin{abstract}
  Following Feynman's prescription for constructing a path integral
  representation of the propagator of a quantum theory, a short-time
  approximation to the propagator for imaginary time, $N=1$
  supersymmetric quantum mechanics on a compact, even-dimensional
  Riemannian manifold is constructed.  The path integral is
  interpreted as the limit of products, determined by a partition of a
  finite time interval, of this approximate propagator.  The limit
  under refinements of the partition is shown to converge uniformly to
  the heat kernel for the Laplace-Beltrami operator on forms. A
  version of the steepest descent approximation to the path integral
  is obtained, and shown to give the expected short-time behavior of the
  supertrace of the heat kernel.

\end{abstract}

\section*{Introduction}
\subsection*{Motivation}
The  path integral has proven an extremely
powerful tool in quantum theory, and, surprisingly, has proven a
powerful heuristic tool in mathematics as well.  Adept use of path
integral reasoning has constructed elegant arguments for ostensibly purely
mathematical propositions whose proofs otherwise were unknown or
required deep and quite different arguments \cite{Witten89a,AJ90,Witten91,BT93,SW94,Witten94}. Blau~\cite{Blau93} reviews several
such arguments and their interpretation as infinite-dimensional
applications of the Matthai-Quillen formalism.
 
Among the simplest are the path integral ``proofs'' of various
versions of the Atiyah-Singer index theorem, including those for the
twisted Dirac and DeRham complexes, the latter of these giving both
the Gauss-Bonnet-Chern (GBC) theorem and the Hirzebruch signature
theorem. These heuristic arguments are due independently to
Alvarez-Gaum\'e \cite{Alvarez83} and Friedan and Windey \cite{FW84},
based on an approach suggested by Witten \cite{Witten82a,Witten82b}.
Focusing on the DeRham complex for simplicity, the argument begins
with a heuristic identifying the path integral for imaginary-time
$N=1$ supersymmetric quantum mechanics (SUSYQM), with the supertrace
of the heat kernel of the Laplace-de Rham operator on forms. An
argument of McKean and Singer \cite{MS67} shows this supertrace is
equal to the signed sum of the Betti numbers, independent of the time
$t.$ On the other hand, the steepest descent approximation, applied
formally to this path integral, equates the small-$t$ limit of the
integral 
to the Pfaffian of curvature, recovering the GBC theorem.

This is arguably the simplest use
of path integral reasoning to arrive at a clearly nontrivial, purely
mathematical conclusion.  As such it is an important benchmark for a
rigorous interpretation of the path integral, in the sense that a
rigorous interpretation of a heuristic reasoning tool ought to offer a
reasonably direct way to translate the heuristic arguments into rigorous
proofs.  In particular, a construction of the path integral for the
supersymmetric propagator that is robust enough to turn the path
integral argument for the GBC theorem into a real proof should represent
an incremental step towards understanding similar arguments
rigorously.  The construction, and the proof of its properties, are the
focus of this paper.

\subsection*{Related results}
The authors' previous work~\cite{FS08} provides a first step by
rigorously constructing the path integral over paths and showing it
agrees pointwise with the heat kernel. The improvement described in
the present paper, which starts with a slightly modified approximate
kernel, is to obtain estimates on the small-$t$ behavior of the path
integral sharp enough to prove that the usual steepest-descent heuristic
 in fact gives the correct short-time approximation to the path integral. 

There is other closely-related work in the literature.  Bismut
\cite{Bismut84a,Bismut84b} uses stochastic techniques with the heat
equation to give a proof of the index theorem in the spirit of the
physics argument. Getzler~\cite{Getzler86a,Getzler86b}, who like
Bismut does not directly construct path integrals, gives an index
theorem proof influenced by these arguments.  Rogers \cite{Rogers87}
uses stochastic techniques to construct an explicit supersymmetric
path integral for the heat kernel on $\RR^n$ with a Riemannian metric which
is Euclidean outside of a bounded region. This suffices to reproduce
the path integral proof of the GBC theorem for arbitrary compact
manifolds, since the argument only depends on the short-time behavior
of the restriction of the heat kernel to the diagonal.  In
\cite{Rogers92a,Rogers92b}, she extends these techniques to prove the
twisted Hirzebruch index theorem, from which follows the full index
theorem.  Andersson and Driver \cite{AD99}, use stochastic techniques
to construct a version of the bosonic path integral on curved space.  The
innovation in the present paper is to make rigorous Feynman's original
time-slicing procedure in constructing the supersymmetric path
integral for the heat kernel on an arbitrary compact Riemannian
manifold, and to recover, from this construction, its short-time
approximation. As noted above, this gives a proof of the GBC theorem.

          \subsection*{Technical Introduction}
The principal result of this paper is Theorem~\ref{th:kinf} which shows
that the fine-partition limit of products of the approximate kernel $K$ defined in
Eq.~\eqref{eq:K-def} converges to the heat
kernel. Arguably the fine-partition limit captures the essence of the
Feynman path integral while  $K$  captures the small-$t$
approximation of the integrand in the path integral representation of
the SUSY QM propagator for imaginary time, so Theorem~\ref{th:kinf} may be
viewed as a rigorous construction of the path integral for the
imaginary-time propagator.  Moreover, Eq.~\eqref{eq:kinf-est} in its
special case 
Eq.~\eqref{eq:kinf-kest}, which relates the large-partition limit
explicitly to the approximate kernel $K$, provides a precise statement of the
steepest-descent approximation to the path integral.  The GBC
 Theorem 
(Thm.~\ref{th:gbc}) is an immediate consequence of this approximation,
because the error term in Eq.~\eqref{eq:kinf-kest}, expressing the
difference between the path integral and its steepest-descent
approximation, is bounded in a norm 
calibrated to send the supertrace of this error term to zero.

The norm is somewhat complicated because it serves
two masters:  On the one hand, the error terms in the
large partition limit of $K$ must be bounded in this norm, which means
it must be designed so that, 
with errors measured in this norm, $K$ almost
satisfies the heat equation, the norm is almost
multiplicative under $*$, and $K$ is almost a contraction,  as
expressed in Eqs.~\eqref{eq:heat-est2}, 
\eqref{eq:l*l}, and 
\eqref{eq:k*l}, respectively.
  On the other hand, as above, it must ensure
that the supertrace of the error terms in Eq.~\eqref{eq:kinf-kest}
must go to zero 
in order for Theorem~\ref{th:gbc} to hold.  More explicilty, 
Eq.~\eqref{eq:heat-est2} requires a bound on error terms of the form
of polynomials times Gaussians, which become more sharply peaked as
$t$ approaches $0$. The norm compares these terms
  near the diagonal to Gaussians and away from the
diagonal to
powers of $t$ (in fact $t^1$ suffices) times a constant.  This would
completely suffice were it not for the second 
consideration; in fact, Eqs.~\eqref{eq:k*l}
and~\eqref{eq:l*l} could be strengthened and simplified.   However,
to get
the behavior under the supertrace the heuristics suggest
not only must the error terms be
bounded by $t$ to a power greater than $1,$ but the norm must take into
account that these errors are forms; ultimately,  the top degree piece
of the error
 must be bounded by a power of $t$ greater than
$n/2.$  This would suggest choosing the norm so that a multiform of
highest degree $k$ and constant in $t$
has norm $t^{-1/2}$ to the power $\max(0,
k-2).$  Such a norm still works  for Eqs.~\eqref{eq:heat-est2}
and~\eqref{eq:l*l}, but, unfortunately, just barely fails for
Eq.~\eqref{eq:k*l}.  So  amend the norm again and instead of $t^{-1/2}$
use $t^{-1/2+\epsilon/n}$  for a sufficiently small $\epsilon.$ This somewhat unintuitive choice
satisfies all the conditions.
\subsection*{Outline}
The structure of the paper is as follows.  Section 1 reviews standard
results from differential geometry and Riemann normal coordinates, in
particular describing approximate solutions to generalized heat
equations.  This will all be used towards the proof of
Eq.~\eqref{eq:heat-est2}.  Section~\ref{grassman} reviews standard facts and
notation of Grassman variables, and defines the norm
used in the proofs of Eq.~\eqref{eq:heat-est2} and of
Eqs.~\eqref{eq:k*l} and~\eqref{eq:l*l} (as well as of
Theorem~\ref{th:gbc}).  Section 3 proves Eq.~\eqref{eq:heat-est},
from which Eq.~\eqref{eq:heat-est2} follows easily.  Sections 4 and 5
prove Eqs.~\eqref{eq:k*l} and~\eqref{eq:l*l}, as well as
Eq.~\eqref{eq:kk-est}, which follows from Eq.~\eqref{eq:heat-est}.
Sections 6 and 7 use only Eqs.~\eqref{eq:k*l}, \eqref{eq:l*l},
and~\eqref{eq:kk-est} to prove Theorem~\ref{th:kinf} showing that the
large partition limit converges to the propagator and
providing bounds on the error.  Finally Section 8 checks that this
construction of the path integral, with the given version of the
steepest descent approximation indeed suffices to obtain the GBC
 Theorem.

\section{Riemann Normal Coordinates and Kernels}
 
The following material on Riemann normal coordinates is
standard\footnote{A readily-accessible version is at \\
  http://users.monash.edu.au/~leo/research/papers/files/lcb96-01.pdf}
and follows straightforwardly
from formulae in~\cite{BGV04}.

Lemma~\ref{lm:dsquared} is
used only in the proof of Proposition~\ref{lm:h-bd} and is
elementary.  Lemma~\ref{lm:approx-heat} gives approximate fundamental
solutions for some generalized heat equations in Riemann normal
coordinates;  it is used in the proof of Proposition~\ref{pr:heat-est}
and is purely computational.

             \subsection{Riemann Normal Coordinates}\label{ss:rnc}

Let $M$ be a fixed compact oriented Riemannian manifold of even
dimension $n.$ Since $M$ is compact there is an injectivity radius such
that for all $y\in M$ the exponential map $\exp_y\colon T_yM \to M$ is
invertible on points within  the injectivity radius of $y$ (with
respect to the Riemann distance).  Refer to the inverse image of
$x\in M$ within the injectivity radius as $\vec{x}_y \in T_yM$ (if the
subscript is obvious omit it).  A choice of orthonormal coordinates
$(\partial_i)_{i=1}^n$ on $T_yM$ then gives coordinates on a patch
around $M$ which are called Riemann normal coordinates. In such coordinates
express the metric
$g_{ij}\parens{\vec{x}}= \parens{\partial_i, \partial_j}_x,$ its
  inverse $g^{ij}\parens{\vec{x}},$ the Levi Civita connection $
  \Gamma_{ij}^k\parens{\vec{x}}=g^{kl}\parens{\partial_l,
    \Nabla_{\partial_i} \partial_j}_x$
and the Levi
  Civita parallel transport from $y$ to $x$ along the unique geodesic
  connecting them $\pt_{ij}\parens{\vec{x}} = \parens{\partial_i,
    \pt^y_x \partial_j}_x$,
 in terms of the curvature
 at the origin $y$
  by the following formulas
 \begin{align}
g_{ij}\parens{\vec{x}}&= \delta_{ij} + \frac{1}{3} R_{ikjl}(0) x^k x^l
+ \OO\parens{\abs{\vec{x}}^3} \label{eq:g-rnest}\\
g^{ij}\parens{\vec{x}}&= \delta^{ij} - \frac{1}{3} \delta^{ii'}R_{ki'l}^{\,\,\,\,\,\,\,j} (0) x^k x^l
+ \OO\parens{\abs{\vec{x}}^3} \label{eq:ginv-rnest}\\
\Gamma_{ij}^k\parens{\vec{x}}&= - \frac{1}{3}
\sbrace{R_{ilj}^{\,\,\,\,k}(0) + R_{jli}^{\,\,\,\,k}(0)}x^l
+ \OO\parens{\abs{\vec{x}}^2} \label{eq:gamma-rnest}\\
\pt_{ij}\parens{\vec{x}}&= \delta_{ij} +\frac{1}{6} R_{ikjl}(0) x^k x^l
+ \OO\parens{\abs{\vec{x}}^3} \label{eq:pt-rnest}
 \end{align}
where $\OO\parens{\abs{\vec{x}}^3}$ refers to an unnamed function which is
bounded by some constant multiple of $\abs{\vec{x}}^3.$ Here the components of the Riemann curvature at $x$ satisfy $R\sbrace{\partial_i, \partial_j} \partial_k =
R_{ijk}^{\,\,\,\,\,\,l} \partial_l\parens{\vec{x}}$ and $R_{ijkl}\parens{\vec{x}}= \parens{\partial_l,
    R\sbrace{\partial_i, \partial_j} \partial_k}_x$.

The  Ricci and scalar curvatures are
\begin{align}
\Ricci_{ij}\parens{\vec{x}}&=
g^{kl}\parens{\vec{x}}R_{kijl}\parens{\vec{x}}\label{eq:ricci-def}\\
\scalar\parens{\vec{x}}&=
g^{ij}\parens{\vec{x}}\Ricci_{ij}\parens{\vec{x}}\label{eq:scalar-def}
\end{align}

\begin{lemma} \label{lm:dsquared}
If $\vec{x}$ and $\vec{y}$ are Riemann normal coordinate expressions
about some point
for $x$ and $y$ which 
are close enough so the distance between them $d(x,y)$ is defined, then
\be
d(x,y)^2 = \abs{\vec{x}-\vec{y}}^2 + \OO\parens{\abs{\vec{x}}^2
  \abs{\vec{y}^2}}.
\ee{dist-est}
\end{lemma}

\begin{proof}
Consider the difference $f(\vec{x}, \vec{y}) = d(x,y)^2 -
\abs{\vec{x}-\vec{y}}^2$ as a function of $\vec{x}$ for fixed
$\vec{y}$. It is smooth,  and, since the exponential preserves length
along geodesics through the origin, $f(0,\vec{y})=0,$ and  the
directional derivative 
of $f$ at $\vec{x}=0$ in the direction $\vec{y}$ is zero as well.
Gauss's lemma ensures the directional derivative at $\vec{x}=0$  in directions
perpendicular to $\vec{y}$ are also $0.$  So $\partial
f/ \partial \vec{x} =0$ at $\vec{x}=0.$ Since this is true for all
$\vec{y},$ $\partial^2 f/\partial \vec{y}^2$ and $\partial^3
f/ \partial \vec{x} \partial \vec{y}^2=0$ at $\vec{x}=0,$ as well, so
$\partial^2 f/\partial \vec{y}^2$ is bounded by some $c\abs{\vec{x}}^2$
for all $x$ and $y$ by Taylor's theorem.  By symmetry, $f$ and
$\partial f / \partial \vec{y}$ are zero when $\vec{y}=0,$ so
integrating $\partial^2 f / \partial \vec{y}^2$ twice against $y$
gives the result.
\end{proof}
             \subsection{Heat Equations and Approximate Kernels}
Of particular interest are approximate solutions to various versions of
the heat equation in Riemann normal coordinates.

\begin{lemma} \label{lm:approx-heat}
In a neighborhood of $0$ in Riemann normal coordinates if
\[\Delta_0= \frac{1}{2}g^{ij}\parens{\partial_i \partial_j -
  \Gamma_{ij}^k \partial_k}\]
and 
\[ K_0\parens{\vec{x};t}= (2\pi t)^{-n/2} e^{ - \frac{\delta_{jk}x^j
    x^k}{2t} + \frac{1}{12} \Ricci_{kl} x^k x^l + \frac{t \scalar}{12}}\]
then 
\be
\frac{\partial K_0}{\partial t} - \Delta_0 K_0 = \OO\parens{t
  + \abs{\vec{x}} + \abs{\vec{x}}^3/t + \abs{\vec{x}}^5/t^2}K_0.
\ee{heat0-est}
Likewise if $f\parens{\vec{x}}$ is twice differentiable with $f$ and all its first and
second derivatives bounded by $C$ and
\[\Delta_1 = \Delta_0 + f\parens{\vec{x}}\]
and
\[K_1\parens{\vec{x};t} = K_0\parens{\vec{x};t}
e^{\frac{t}{2}\sbrace{f\parens{\vec{x}}+f\parens{0}}}\]
then 
\be
\frac{\partial K_1}{\partial t} - \Delta_1 K_1 = \OO\parens{t
  + \abs{\vec{x}} + \abs{\vec{x}}^3/t + \abs{\vec{x}}^5/t^2 + Ct +
  C\abs{\vec{x}}^2 + C^2 t^2}.
\ee{heat1-est}

Finally if $h_k^l$ is for $0 \leq k,l \leq n$ a collection of
functions of $\vec{x}$ 
all bounded by $D,$ then defining
\[\Delta_2 = \Delta_1 + h_k^l x^k\partial_l + h_k^k/2 \]
and
\[K_2\parens{\vec{x};t}= K_1\parens{\vec{x};t} e^{- h_k^l \delta_{ll'}x^k
x^{l'}/2}\]
then 
\eqa{heat2-est}{
\frac{\partial K_2}{\partial t} - \Delta_2 K_2 &= \OO\Big(t
  + \abs{\vec{x}} + \abs{\vec{x}}^3/t + \abs{\vec{x}}^5/t^2 + Ct +
  C\abs{\vec{x}}^2 \nonumber \\
&\qquad + C^2 t^2 + D\abs{x} + D\abs{\vec{x}}^3/t + CD\abs{x} t
  + D^2\abs{x}^2 \Big) K_2.
}

\end{lemma}

\begin{proof}
Here, and in the remainder of this proof, write $x_j$ for
$\delta_{jj'}x^{j'}.$  For Eq.~\eqref{eq:heat0-est}, observe that 
\be
\frac{\partial K_0}{\partial t} = \sbrace{-\frac{n}{2t} +
  \frac{ x^j
  x_j}{2t^2} + \frac{\scalar}{12}} K_0
\ee{k0dt}
and 
\be
\partial_i K_0 = \frac{\partial K_0}{\partial x^i} =
\sbrace{-\frac{x_i}{t} +
  \frac{\Ricci_{ik}x^k}{6} + \OO\parens{t + \abs{\vec{x}}^2}}K_0.
\ee{partialk0}

Therefore (implicitly symmetrizing both sides, since this term
occurs in $\Delta$ only symmetrically)
\begin{align*}
&\partial_i \partial_j K_0= \Bigg[-\frac{\delta_{ij}}{t} + \frac{\Ricci_{ij}}{6}+
   \frac{x_{i} x_{j}}{t^2} \\
& \qquad - \frac{\Ricci_{jk}x^k
    x_{i}}{3t} + 
  \OO\parens{t + \abs{\vec{x}} + \abs{\vec{x}}^3/t}\Bigg]K_0
\end{align*}
and
\[\Gamma_{ij}^k \partial_k K_0=
\sbrace{-\Gamma_{ij}^k\frac{x_{k}}{t} +
  \OO\parens{\abs{\vec{x}}}} K_0.\]
Using Eqs.~\eqref{eq:gamma-rnest} and~\eqref{eq:ginv-rnest} to expand
the metric and Christoffel symbol in Riemann normal coordinates, and dropping low enough order terms
\eqa{delta-k0}{
&\Delta_0 K_0 = \Bigg[-\frac{g^{ij}\delta_{ij}}{2t} +
\frac{g^{ij}\Ricci_{ji}}{12}+ \frac{ g^{ij}
    x_{i} x_{j}}{2t^2} -
  \frac{g^{ij} \Ricci_{jk} x^k x_{i}}{6t} \nonumber\\
&\qquad  + \frac{g^{ij}\Gamma_{ij}^k x_{k}}{2t} + \OO\parens{t + \abs{\vec{x}}+\abs{\vec{x}}^3/t}\Bigg]K_0 \nonumber\\
&= \Bigg[-\frac{n}{2t} + \frac{ R_{kil}^{\,\,\,\,\,\,i} x^k x^l}{6t} + \frac{\scalar}{12} + \frac{  x^k x_{k}}{2t^2} -
  \frac{\Ricci_{kl} x^k x^l}{6t} + \frac{\delta^{ii'}\Gamma_{ii'}^k x_k}{2t}\nonumber\\
& \qquad + \OO\parens{t
    + \abs{\vec{x}}+ \abs{\vec{x}}^3/t + \abs{\vec{x}}^5/t^2}\Bigg]K_0 \nonumber\\
&= \sbrace{-\frac{n}{2t} + \frac{ x^k x_k}{2t^2} + \frac{\scalar}{12} + \OO\parens{t
    + \abs{\vec{x}}+\abs{\vec{x}}^3/t+ \abs{\vec{x}}^5/t^2}}K_0
}
where we have used
Eqs.~~\eqref{eq:gamma-rnest}
and~\eqref{eq:ricci-def} in the last line.  Subtracting
Eqs.~\eqref{eq:k0dt} and~\eqref{eq:delta-k0} gives the result.  

For Eq.~\eqref{eq:heat1-est}, notice $\partial K_1/\partial t$ is
given by all the terms in Eq.~\eqref{eq:k0dt} multiplied by $K_1$
instead of $K_0$ plus in addition
\[ \frac{f(0)+f\parens{\vec{x}}}{2} K_1.\]
Also 
\[\partial_i e^{\frac{t}{2}\sbrace{f(0)+f\parens{\vec{x}}}}
=\frac{t}{2}\frac{\partial f\parens{\vec{x}}}{\partial x^i} e^{\frac{t}{2}\sbrace{f(0)+f\parens{\vec{x}}}}
\] 
so estimating $\partial_i \partial_j
e^{\frac{t}{2}\sbrace{f(0)+f\parens{\vec{x}}}}$ as $\OO\parens{Ct+ C^2t^2}$
and estimating $\partial_i K_0 \partial_j
e^{\frac{t}{2}\sbrace{f(0)+f\parens{\vec{x}}}}$  with the help of
Eq.~\eqref{eq:partialk0} we see $\Delta_1 K_1$ contains all the terms in Eq.~\eqref{eq:delta-k0}
(times $K_1$ instead of $K_0$) plus
\[ \sbrace{f\parens{\vec{x}} - \frac{x^{i}}{2}\frac{ \partial
      f\parens{\vec{x}}}{\partial x^i}  + \OO\parens{ 
    tC + t^2C^2}   } K_1.\]
Noticing that 
\[\frac{f(0)+f\parens{\vec{x}}}{2} = f\parens{\vec{x}} - \frac{x^{i}}{2}\frac{ \partial
      f\parens{\vec{x}}}{\partial x^i}  +
    \OO\parens{C\abs{\vec{x}}^2}\]
by Taylor's Theorem gives the result. 

For Eq.~\eqref{eq:heat2-est}, there are no additional terms to
$\partial K_2/\partial t$ beyond those in $\partial K_1/\partial t$
except multiplied by $K_2$ instead of $K_1.$ The additional term in
$\partial_j K_2$ is (writing $h_{jk}$ for $\delta_{jj'} h^{j'}_k$)
\[\parens{-\frac{1}{2}\sbrace{h_{jk} + h_{kj}} x^k + D\abs{\vec{x}}^2 }K_2\]
So the additional terms in $\partial_i \partial_j K_2$ are
(symmetrizing on $i$ and $j$) 
\[
\sbrace{\frac{h_{kj}  x^k x_{i} + h_{jk}
     x^{k} x_{i} }{t} - h_{ij}  +
  \OO\parens{D\abs{\vec{x}} + CDt\abs{\vec{x}}^2 + D^2 \abs{\vec{x}}^4}
} K_2 
\]
which means that 
the additional terms in $\Delta_2 K_2$ consist of
\begin{align*}
&\Bigg[\frac{g^{ij} h_{jk}   x^k x_{i} }{t}  -
  \frac{1}{2}g^{ij}h_{ij} -
  h_{kl}  x^{k} \frac{x^l}{t} + h_k^k/2\\
&\qquad  + 
  \OO\parens{D\abs{\vec{x}} +D \abs{\vec{x}}^3/t+ CD\abs{\vec{x}}^2t + D^2\abs{\vec{x}}^4}
\Bigg]K_2
\end{align*}
from which the result follows.
\end{proof}


\section{Grassman Variables and the Norm}  \label{grassman}

The first subsection on Grassman variables is completely standard and
can be found for example in \cite{MQ86,Rogers92a}.  Expressing
operators on forms in terms of kernels in Grassman variable is also in
\cite{Rogers92a}. 
 The norm $\abs{\,\cdot\,}_t$ is designed primarily
to keep track of the error 
terms in the large partition limit in such a fashion that the
supertrace of the error terms goes to zero as required for 
Gauss-Bonnet-Chern theorem.  Apart from various baroque
analytic additions, the idea of the norm and the rank of terms is a
fairly natural  
scaling
wherein the spatial variable $\vec{y}_x$
scales as $t^{1/2} $ and the Grassman variable $\psi$ scales as
$t^{-1/2}.$  Lemma~\ref{lm:tnorm-facts} collects facts about the norm
used elsewhere, but the key fact is the first, which uses the presence
of the Gaussian to translate powers of the spatial variable into
powers of $t^{1/2}$ near the diagonal, and an exponentially damped term
 away from the diagonal. This is crucial to the proof of Proposition~\ref{pr:heat-est}.

             \subsection{Grassman Variables}\label{ss:grassman}

If $V$ is a vector space with basis $v_1, \ldots v_n$ then a $k$-linear function  of $k$ variables $w_i= w_i^j v_j$
for $1 \leq i \leq k$ takes the form
\[f\parens{w_1, \ldots, w_k}= C_{j_1, j_2, \ldots j_k}w_1^{j_1}
w_2^{j_2}\cdots w_k^{j_k}.\]
The antisymmetrization of this multilinear function is the element of
$\Lambda^* V$ given by 
\[A{f}\parens{v_1, \ldots, v_k}= \sum_\sigma \frac{(-1)^\sigma }{k!} f\parens{w_{\sigma(1)},
  \ldots, w_{\sigma(k)}}\]
which is again $f$ if it was already antisymmetric. This map
from multilinear functions to antisymmetric ones is manifestly basis
independent, but using the basis  gives a simple representation
of $\Lambda^* V$ as the algebra generated freely by $\psi^1, \psi^2,
\ldots , \psi^n$ subject only to the antisymmetry relation $\psi^j
\psi^j = - \delta \psi^j \psi^i.$  The antisymmetry map is then given
naturally by writing $\psi= \psi^1 + \cdots + \psi^n$ and interpreting
$Af$ as
\[f(\psi) \equiv C_{j_1, j_2, \ldots
  j_k}\psi^{j_1} \cdots \psi^{j_k}.\]
Call such an $f$ an antisymmetric polynomial in the
Grassman variable $\psi$ with values in $V.$ 

To represent endomorphisms on $\Lambda^* V,$ define maps
$\partial/ \partial \psi^j$ for $1 \leq j \leq n$ by 
\[ \frac{\partial \,}{\partial \psi^j} \psi^j \psi^L= \psi^L
\qquad \frac{\partial \,}{\partial \psi^j} \psi^L=0\]
if $\psi^L= \psi^{l_1} \cdots \psi^{l_k}$ and $j \neq l_i$ for all
$i.$  It is easy to confirm that multiplication operators
$\psi^j$ and the differentiation operators $\partial/ \partial \psi^j$
generate all endomorphisms.  In fact the maps $\partial/\partial \psi^j$  generate a free
antisymmetric algebra as well, so the same convention applies to
interpret multilinear functions $g\parens{\partial/\partial \psi}$ of
these.  Expressions of the form  $f\parens{\psi}
g\parens{\partial/ \partial \psi}$ define endomorphisms (and in fact
span all endomorphisms). 
Operators that preserve form degree
are spanned by those products in which $f$ and $g$ are homogeneous
of the same monomial degree.  Filter the form-degree-preserving endomorphisms by
this common monomial degree, and call this the \emph{degree} of the endomorphism. An
endomorphism of degree $p$ is zero on all elements of $\Lambda^*V$ of
degree less than $p$ and nonzero on at least one element of degree
$p.$

Given a choice of degree $n$
element of $\Lambda^* V,$ any operator $\kk$ on $\Lambda^* V$ determines a
kernel; that is, an antisymmetric
polynomial $K\parens{\psi_x, \psi_y}$ in two Grassman variables
$\psi_x, \psi_y$ valued in $V$ (which of course anticommute with each
other). Explicitly, applying $\kk$ to $f(\psi)$ is
\[ (\kk f)(\psi) = \bint K\parens{\psi, \psi_y} f\parens{\psi_y} d\psi_y\]
where the Berezin integral $\bint f(\psi_y) d\psi_y$ returns the
coefficient of the chosen top element of $\Lambda^* V$ in $f.$ Such
kernels in turn admit a sort of Fourier
transform, in terms of a Grassman variable $\rho$ valued in
$V^*$ and an antisymmetric polynomial in three variables
$k\parens{\psi_x, \rho, \psi_y}:$
\[K\parens{\psi_x, \psi_y} =  \bint e^{i\bracket{\rho, \psi_x - \psi_y}} k\parens{\psi_x, \rho,
    \psi_y} d\rho\]
where the Berezin integral is with respect to the element of
$\Lambda V$ dual to the chosen element of $\Lambda^* V.$  If
$k(\psi_x, \rho, \psi_y)$
is of the form $f(\psi_x) g(\rho) h(\psi_y)$, the operator $\kk$ is
simply $\kk = f(\psi) g(i \partial/ \partial \psi) h(\psi).$  Degree-preserving endomorphisms are spanned by homogeneous $k$ in which the
degree of $\rho$ equals the sum of degrees of $\psi_x$ and $\psi_y,$
and the degree of the endomorphism is the degree of $\rho.$

             \subsection{Bundles}\label{ss:bundle}

With the help of the mild subterfuge of the previous
subsection,  sections of all bundles in this
paper may be expressed using the ``function of variables'' notation familiar from
freshman calculus.  Usually these sections will also depend on a
positive real
parameter $t$ (or sometimes $t_j$ or $\tau_j$), which  will appear in
the list of parameters at the end, after a semicolon, to emphasize its
special role.  For the entire paper \emph{except}
Corollary~\ref{cr:kinf} this $t$ will be bounded above by an arbitrary
constant, and all other constants, explicit and implicit, will depend on this
bound as well as on the manifold $M.$ 

For instance, $f(x,\psi_x)$ and
$f(x,\psi_x;t)$ represent sections of the  bundle $\Lambda^*T_M$  over
$M,$ i.e. elements of $\Gamma \parens{\Lambda^*T_M},$ or forms. This notation is meant to indicate that at each $x$,
$\psi_x$ is a Grassman variable taking values in $T_xM,$ and
$f(x,\psi_x)$ is a multilinear function of this variable which thus
represents an element of $\Lambda^* T_xM.$  One useful aspect of this
particular notation is that, assuming the Berezin integral uses the
measure coming from the metric on $T_xM,$ $\int \bint f(x, \psi_x)
d\psi_x \, dx$ represents the integral of the top form component of
$f,$ regardless of the metric.  That is, it represents the usual
integral on forms.

The arguments in the paper  are concerned with sections of a particular bundle
over $(x,y)  \in M \times M$ which will be most conveniently discussed in
two distinct representations. (In fact, the representations only agree
near the diagonal, but that suffices.) The first bundle is $\Lambda^* T_x
M \tensor \Lambda^* T_yM,$ whose sections in
$\Gamma \parens{\Lambda^*T_xM} \tensor \Gamma \parens{\Lambda^*T_yM}$ can be represented by functions $L(x,y;\psi_x,
\psi_y)$
 in the obvious fashion.  Define $E_{x,y}$ to be the set of all such
 sections which are degree preserving (that is, the kernel is a sum of
 monomials in which the degree of $\psi_x$ and the degree of $\psi_y$
 add to $n$). As in the previous subsection, such sections represent kernels of operators on forms, via the $*$
product
\be
(\lop f)(x,\psi_x)=\sbrace{L*f}(x,\psi_x)= \int \bint L(x,y,\psi_x, \psi_y) f(y, \psi_y)
d\psi_y\, dy.
\ee{*-def}
Usually the kernels will only be nonzero when $x$ and $y$ are within
the injectivity radius of each other, so that the parallel transport
map $\pt^x_y$ from $T_xM$ to $T_yM$ is defined.  
These kernels will take the form
\[\bint e^{i\bracket{\rho,\pt^x_y \psi_x - \psi_y}} k(x,y,\psi_x,
\psi_y, \rho) d\rho.\]
The simplest example would be the kernel
\[P(x,y,\psi_x, \psi_y)= \bint e^{i \bracket{\rho, \pt^x_y \psi_x -
    \psi_y}} d\rho\]
($P=0$ when $\pt^x_y$ is not defined).

If $x$ and $y$ are sufficiently close, then  $L \in E_{x,y}$
  determines a
section $\overline{L} \in \overline{E}_{x,y}$
 of the bundle $\End_0\parens{\Lambda^*
  T_xM} \tensor \RR_y$ defined by
\[L(x,y,\psi_x, \psi_y)= \overline{L}(x,y) P(x,y,\psi_x,\psi_y).\]
Here the subscript $0$ is a reminder that the endomorphisms are
degree-preserving, and the subscript on $\RR_y$ indicates the bundle is
only over the first factor of $M.$  Thus 
\[(\lop f)(x,\psi_x)= \int \sbrace{\overline{L}(x,y) \ptxystar f}(\psi_x) dy.\]
Here $\ptxystar : \Lambda^* T_yM \rightarrow \Lambda^* T_xM$ is  the
dual to the parallel transport operator, extended
antisymmetrically. As an operator on functions of $\psi_y$, $\ptxystar$ is the
pullback by $\pt^x_y$.
             \subsection{The Norm on the Bundle}\label{ss:norm}
Fix for the rest of the paper an arbitrary $\epsilon$ with $1/2 >
\epsilon > 0.$  All explicit and implicit constants in what follows will
depend on $\epsilon.$ 
Define a norm $\abs{\, \cdot \,}_t$ depending on a positive real $t$
on $\End_0(T_xM)$ as follows: 
\begin{definition} In terms of an orthonormal basis,  
\be
\abs{\psi^{j_1} \cdots \psi^{j_k} \frac{\partial\,}{\partial
    \psi^{l_1}} \cdots \frac{\partial\,}{\partial
    \psi^{l_k}} }_t = \begin{cases} 1 & k \leq 2\\
t^{(k-2)(-1/2+\epsilon/n)} & k \geq 2.\end{cases}
\ee{norm-def1}
\end{definition}
This is the inner product norm times a power of
$t$ depending only on the degree.  If $\overline{L} \in
\overline{E}_{x,y}$ define $\abs{\overline{L}}_t$ to be the function
of $x$ and $y$ given by this norm on each fiber, and for $L \in
E_{x,y}$ define $\abs{L}_t= \abs{\overline{L}}_t.$

There is a real-valued kernel $H$ of particular importance,
because, where defined, it agrees with the flat-space heat operator; namely
\be
H\parens{x,y;t}= (2\pi t)^{-n/2} e^{-\frac{\abs{\vec{x}_y}^2}{2t}}
\ee{h-def}
if $x$ and $y$ are within the injectivity radius of each other and $H=0$
otherwise.  
The errors introduced by using approximate
heat kernels on forms will generally be $t$-dependent sections which take the form of $H$ times a section well-behaved as $t$ goes to zero.

Suppose $L \in E_{x,y}$ is  a $t$-dependent kernel of degree $k$,
and
suppose there are real numbers $C > 0,$ $r,$ and $s$ such that
$\abs{L} \leq Ct^r \abs{\vec{y}_x}^s$ for all $x$ and $y$ within the
injectivity radius.  Then define $L$ to be of {\em rank} $p$ where $p =
\min(2r+s, 2r+s + 2 - k)$.  The sum of terms of rank $p$ of different
degrees is still called rank $p.$  The same definition applies to
elements of $\overline{E}_{x,y}$ and the rank of $L$ is the rank of $\overline{L}.$

For instance $\Ricci\parens{\vec{y}_x,
    \vec{y}_x},$ $t\scalar,$ and $t\parens{\rho, R\sbrace{\psi_x,
      \psi_x}\rho}$ all are  rank $2,$ while $\bracket{\rho,
    R\sbrace{\vec{y}_x, \psi_x} \vec{y}_x}$ is rank $3.$ The
  significance of rank is explained in the next lemma:

\begin{lemma} \label{lm:tnorm-facts}
\begin{alist}
\end{alist}

\begin{alist}
\item 
If $\overline{L} \in \overline{E}_{x,y}$
 is of rank $p$ or less and is
a smooth function of $t$ times $t^c,$ where $c \geq -2,$ 
  there is an $A > 0$ such that 
\be
\overline{L}H= \overline{F}_1 H + \overline{F}_2 t
\ee{rank-bd}
where $\overline{F}_i \in \overline{E}_{x,y}$ and $\abs{\overline{F}_i}_t\leq At^{(p-1)/2
   +\epsilon}.$ 
\item If $\overline{L}$ is a smooth function of $t_1 \leq t$ of degree at most
  $1$ and rank $2$,  or degree $2$ and rank $2$ by virtue of being a
 multiple of $t_1,$ then 
there are $B,D > 0$ such that if $\overline{F} \in \overline{E}_{x,y}$
\eqa{almost-contract}{
\abs{\overline{L}H(t_1)\overline{F}}_t & \leq Bt_1^{\epsilon/n}\abs{\overline{F}}_t\sbrace{H(t_1) +
Dt_1} \nonumber \\
\abs{\overline{F}\overline{L}H(t_1)}_t & \leq Bt_1^{\epsilon/n}\abs{\overline{F}}_t\sbrace{H(t_1) +
Dt_1}.
}
\item
If $t<t'$ then $\abs{F}_{t'} \leq \abs{F}_t. $ 
\item There is a  $C > 0$ such that for all $\overline{F},\overline{G}
  \in \overline{E}_{x,y}$
\be 
\abs{\overline{F}\overline{G}}_t \leq Ct^{-1} \abs{\overline{F}}_t \abs{\overline{G}}_t.
\ee{norm-prod}
\end{alist}
\end{lemma}

\begin{proof}
\begin{alist}
\item It suffices to prove the result for $\overline{L}$ of rank
  exactly $p,$ and degree $k \geq 0.$  The proof of Eqs.~\eqref{eq:rank-bd}
 and~\eqref{eq:almost-contract}, as well as of Eq.~\eqref{eq:hh-bd}
 below, rely on the following trick for converting powers of
 $\vec{y}_x$ into powers of $t$ in the presence of the Gaussian $H.$  Choose $\delta,C>0,$ and
 consider separately the case $\abs{\vec{y}_x}> Ct^{1/2-\delta}$
 and $\abs{\vec{y}_x} \leq Ct^{1/2-\delta}.$  In the former case
 $H(x,y;t) =\OO\parens{t^{-n/2}} e^{C^2 t^{-2\delta}/2}.$  This is
 exponentially damped, and thus can be bounded by some multiple of
 any power of $t$ desired; in particular,
 $\overline{L}H=\overline{F}_2 t$
 with $\overline{F}_2$
 bounded as in the statement of
 Eq.~\eqref{eq:rank-bd}.  

In the case $\abs{\vec{y}_x} \leq
 Ct^{1/2-\delta},$  since $\overline{L}$ is rank $p$, it is
 bounded by a multiple of $t^r \abs{\vec{y}_x}^s$ where $p=\min(2r+s,2r+s+2-k).$ The bound on $\abs{\vec{y}_x}$
 implies $\overline{L}$  is bounded by a multiple
 of $t^{(2r+s)/2-s\delta}.$  Thus $\abs{\overline{L}}_t$ is bounded
by a multiple of $t^{p/2-s\delta}.$  By assumption, $r \geq -2$, so
$s$ is bounded
(by $ p+4$ if $k \leq 2$ and by $p+k + 2$ otherwise), and therefore,
given $1/2 > \epsilon > 0$, 
$\delta$ can be chosen small enough that in fact $\abs{\overline{L}}_t$ is bounded
by a multiple of $t^{(p-1)/2+\epsilon}$.

\item First notice multiplying two elements of
 $\End_0(\Lambda^* TM)$ of fixed degree together gives a product whose degree
 is the sum of their degrees, so that multiplying a degree  $k$ term
 (smooth in $t$) times
 $\overline{F}$
 multiplies the $t$-norm by a multiple of
 $t^{-k(1/2-\epsilon/n)}.$

Again the result is easy if $\abs{\vec{y}_x} >Ct_1^{1/2-\delta},$ for
some $C$ and $\delta.$ If instead $\abs{\vec{y}_x} \leq
Ct_1^{1/2-\delta},$ then in the case that $\overline{L}$ is degree at
most $1$ and rank $2$, its smoothness implies $r \geq 0$ so it is bounded by a multiple of $t_1^{1-2\delta}$.
Under multiplication it thus
multiplies the $t$-norm of $F$ by a constant times $t_1^{1-2\delta}
t^{-1/2+\epsilon/n} \leq t_1^{1/2 +
  \epsilon/n - 2\delta},$ which with the right choice of $\delta$ is
bounded by a multiple of $t_1^{\epsilon/n}.$  In the case where  $\overline{L}$ is a
multiple of $t_1$ times a degree $2$ term, then upon multiplication it
multiplies the $t$ norm of $\overline{F}$
by a  multiple of $t_1 t^{2(-1/2+\epsilon/n)}\leq t_1^{\epsilon/n}.$  
\item Trivial
\item Multiplication in  $\End_0(\Lambda^* TM)$ adds degree, and clearly  is bounded with
 respect to any fixed norm.  If the degrees $k$ and $j$ are both at
 least $2,$ then the $t$-norm of the product is a multiple of
 $t^{(k+j-2)(-1/2+\epsilon/n)}$ 
which is a multiple of  $t^{2(-1/2+\epsilon/n)}$
times the $t$-norms of the factors, so in particular a multiple of
$t^{-1}$ times the product of the $t$-norms.  If either degree is less
than $2,$ the same is true with a higher power of $t,$ and therefore
in general it is a multiple of $t^{-1}$ times the product of
$t$-norms.
\end{alist}
\end{proof}


             \subsection{The supertrace}\label{ss:str}

The matrix supertrace $\str$ of an  operator $\kk \in \End_0\parens{\Lambda^*
V}$ is the trace of the restriction of $\kk$ to even forms
minus the trace of its  restriction to odd forms.  In terms of
Grassman variables, the kernel of $\kk$ is $K(\psi_x, \psi_y) = K_{IJ} \psi_x^I \psi_y^J$ and the trace of
its restriction to forms of degree $d$ is, in an orthonormal basis, 
\[
\tr \l. \kk \r|_d= \sum_{|I|=d} K_{I\bar{I}}
(-1)^{\bar{I}I},
\]
where $\bar{I}$ is formed from the complement of $I$ in $\{1, \dots,
n \}$, and the last factor is the sign of the
permutation. This follows from noticing $(-1)^{\bar{I}I}
\psi^{\bar{I}}$ is the  dual basis vector to $\psi^I$. On the other hand
\[
\bint K(\psi, \psi) \, d\psi = \sum_{IJ} K_{IJ} \bint \psi^I \psi^J \, d\psi
\]
The integrals on the right vanish, unless $J = \bar{I}$, so the sum
reduces to $\sum_I K_{I\bar{I}} (-1)^{I\bar{I}}$. Since $(-1)^{I\bar{I}}
    = (-1)^{|I||\bar{I}|}(-1)^{\bar{I}I}$, decomposing this sum
according to degree and comparing with the sum for $\tr \l. \kk
\r|_d$, shows
\be
\bint K(\psi, \psi) \, d\psi = \str \kk
\ee{eq:str}
since $n$ is even.
\begin{lemma} \label{lm:str-est}
If $F$ is the kernel of an operator in $\End_0\parens{\Lambda^*
V}$ and $\abs{F}_t \in \OO\parens{t^{p+1+\epsilon}}$, then
\[
\str (F) \in \OO(t^{p + n/2 + 2 \epsilon/n}).
\]
\end{lemma}
\begin{proof}
From the above characterization of the supertrace, $F = \str (F)
\psi^1 \cdots \psi^n \frac{\partial \,}{\partial \psi^1} \cdots
\frac{\partial \,}{\partial \psi^n}$ plus lower-degree terms. 

The condition on $F$ and the definition in Eq~\ref{eq:norm-def1} of
the norm on  $\End_0\parens{\Lambda^*V}$ thus require
\[
\str (F) t^{(n-2)(-1/2 + \epsilon/n)} = \str(F) \abs{\psi^1 \cdots \psi^n \frac{\partial \,}{\partial \psi^1} \cdots
\frac{\partial \,}{\partial \psi^n}}_t \leq \abs{F}_t = \OO\parens{t^{p+1+\epsilon}}
\]
from which the lemma immediately follows.
\end{proof}

\section{The Approximate Kernel}

For background on the heat operator see \cite{BGV04}.
Proposition~\ref{pr:heat-est}, which says $K$ is an approximate
solution to the heat equation with error terms that are Gaussians
times errors bounded in the $t$-norm, is entirely computational.  

\subsection{The Approximate Kernel} \label{ss:ak} Richard Feynman
\cite{Feynman42}, building on ideas of Dirac \cite{Dirac33},
identified the short-time approximation of the kernel of the
propagator for quantum mechanics and, by extension, any quantum theory
with a certain coefficient times the exponential of $i/\hbar $ times
the classical action.  
  The full propagator should be the limit of the
product of many of these short-time approximations, which can be readily
viewed (sweeping the infinite product of the aforementioned
coefficient into the measure) as an integral over all possible paths
of the exponential of $i/\hbar$ times  
the action of the path.  This is the path integral formulation of a
quantum theory.  

Feynman was quite clear there were two
ambiguities in his formulation; namely, the specification of the
coefficient and the existence of the limit. 
 The latter is a
straightforward-enough question, and its answer will take up the bulk
of the paper.  The former is subtle and requires some discussion.

The path integral formalism suggests the propagator for a quantum
theory (and by extension the quantum Hamiltonian) is completely
determined by the action functional plus whatever goes into defining
the putative measure on the space of paths, which one might hope would
involve no more than kinematic aspects of the theory.  If true this
would be quite a remarkable claim, because purely classical data would
determine the quantum Hamiltonian.  In the canonical approach,
quantizing the classical Hamiltonian typically requires a choice of
how to order products of position and momentum operators; the
classical theory alone does not resolve such operator-ordering
ambiguities.

From the Feynman integral point of view the  choice  of operator
orderings appears
the coefficient discussed above.  More specifically, Feynman proposed
approximating the kernel of the propagator by
\[K(x,y;t)= C(x,y;t) e^{i S(x,y;t)/\hbar}\] where $x$ and $y$ are
position variables in the initial and final states and $S$ is the
classical action.  To have a hope of converging under refinements,
this approximate kernel itself must agree with the kernel of the
propagator to first order in $t,$ which means that its derivative at
$t=0$ should give the kernel of the quantum Hamiltonian operator.  The
situation already becomes subtle in bosonic quantum mechanics on
curved space.  There the classical Hamiltonian is $p_i g^{ij}(x) p_j,$
so, promoting the $p$'s to derivatives acting on a spaces of functions
of $x$, the quantum Hamiltonian should be the Laplace-Beltrami operator
plus terms resolving operator-ordering ambiguities. These would
involve one or no spatial derivatives as well as traces of the
curvature and higher order terms in the metric.

There are several ways to chose $C$ and thus specify a short-time
approximation to the propagator.  Geometric quantization
\cite{Woodhouse92} (specifically the BKS pairing) builds $C$ out of
what is essentially a Jacobian factor for the symplectic form
\cite{DeWitt57}, yielding the Van Vleck determinant. Assuming the
Laplace-Beltrami operator is the exact quantum Hamiltonian, the WKB
approximation gives the same $C$ (in the normalization of this paper,
$C=e^{\Ricci\parens{\vec{y}_x, \vec{y}_x}/12}$ to the appropriate
order).  The presence of this term in the approximate kernel proves to
be essential to get appropriate convergence to the kernel of a quantum
Hamiltonian which is the Laplace-Beltrami operator plus $\scalar/6,$
where $\scalar$ is the scalar curvature.  Various other
regularizations of the path integral yield Laplace-Beltrami plus other
multiples of the scalar curvature \cite{Fulling96}.  Aesthetic
considerations lead mathematicians to hope and expect that the
Laplace-Beltrami operator is the ``correct'' Hamiltonian on the nose.
There is no realistic hope of determining which of these nature
prefers, as it seems unlikely the difference between these
Hamiltonians could be accessible to any foreseeable experiment
\cite{Schulman81}.

In the presence of all this confusion, the authors have abandoned the
effort to offer a coherent procedure for choosing $C$ based only on
the action and the kinematics of the theory, and will choose it based
on the desired Hamiltonian (which in the case at hand is the
Laplace-de Rham operator on forms).
  That is, we choose $C$ so that
$Ce^{-S(x,y;t)}$ is a solution to the Laplace-de Rham heat equation to
zeroth order in $t$ (and has total integral $1$ to zeroth order). 
 Inspection
of the resulting approximate kernel in Eq.~\eqref{eq:K-def} shows the
Ricci curvature term corresponding to the Van Vleck term
 in the bosonic case, and a scalar curvature term which
 resolves operator-ordering ambiguities coming from both the
bosonic piece and the four fermi term.  

One final technical point is in order.  In general the metaplectic
correction suggests replacing the
naive notion of the Hilbert space in quantum theory, which is to say
functions on position space, with half-densities or half-forms on this
space \cite{Woodhouse92}.  In the bosonic case this appears in the Van Vleck determinant,
which is naturally viewed as a half-density on the product space.  For
supersymmetric quantum mechanics the naive Hilbert space is forms on
the manifold, which corresponds (see Section~\ref{grassman}) to the space of
functions of a Grassman variable.  Because bosonic and Grasmannian
determinants cancel, half densities on this space can be naturally
identified with the naive space of functions.  The upshot of this
observation is that the remainder of the paper will ignore the
metaplectic correction and half-densities.

With the choice of $C$ determined as above, the small $t$ approximation to the
propagator of $N=1$ imaginary time supersymmetric quantum mechanics in
curved space is given by the kernel
\eqa{K-def}{
& K(x,y,\psi_x,\psi_y;t) =  \bint H\exp\Big[
-\frac{t\scalar }{6} + \frac{1}{12}
\Ricci\parens{\vec{x}_y, \vec{x}_y} 
\nonumber \\
& \quad +  i \bracket{ \rho, \pt^x_y\psi_x- \psi_y} + 
\frac{t}{8} \parens{\rho, R\sbrace{\psi_y, \psi_y} \rho} + \frac{t}{8} \parens{\ptxystar\rho, R\sbrace{\psi_x, \psi_x} \ptxystar\rho}
\Big]d \rho
}
when $x$ and $y$ are close enough to be within the injectivity radius
of each other (so  $K$ is invertible in this region) and zero otherwise.  
All curvature tensors are understood to be evaluated at $y$ except for the
last $R$ which is evaluated at $x.$

The analytic properties of $K$ crucial in what follows will also hold
for any kernel $K'$ which differs from $K$ by a term
which in Riemann normal coordinates around $y$ (that is in terms of
$\vec{x}_y$) is $H(x,y;t)$ times a term of rank  $3$
(see \S\ref{ss:norm}).
Accordingly, write $K \sim K'$ if the two kernels differ by such terms.

Observe that $K=(P+L)H,$ where $P$ is the kernel of the parallel transport operator,
and $L$ is a sum of terms each of which has the form a positive power of $t$, times a power of $\vec{x}_y$, times a product of
$\psi_x, \rho, \psi_y$, wherein the sum of the power of $\vec{x}_y$
and twice the power of $t$ is at least the degree of the $\psi\rho$
factors.  As a consequence, $L$ times any term of rank $p$ is again
rank $p.$  Since $P$ does not affect rank, any
rank $3$ term times $K$ is equal to some rank $3$ term times $H.$ 

For example 
the $\Ricci$ and $\scalar$ terms are
of rank $2,$  and  replacing them with the same terms
interpreted at $x$ (with $\vec{y}_x$ replacing $\vec{x}_y$), changes $K$ by terms of rank $3$ 
times $H$
(since $\Ricci$ and $\scalar$ are both continuously
differentiable). Thus this change produces a kernel equivalent to $K$ under $\sim$.
Likewise, changing variables from $\rho$ to $\ptxystar \rho$ (since $\pt$ is
measure-preserving there is no Jacobian factor) shows that, up to  $H$ times
terms of rank $3$,  $K$ is symmetric with respect to $x$ and $y.$

In fact, up to such terms, $K$ has a simple expression
in Riemann normal coordinates centered at $y$.  Notice 
\[\bracket{\rho, \pt^x_y \psi_x - \psi_y}= \bracket{\rho, \psi_x -
  \psi_y} + \frac{1}{6}\bracket{\rho, R\sbrace{\vec{x}_y, \psi_x}
  \vec{x}_y} + \OO\parens{\abs{\vec{y}_x}^3 \psi_y^i \rho_j}\]
the error term being  rank $3.$
 Likewise, in the term
with $\ptxystar \rho $, the $\ptxystar$ can be
removed at the cost of a rank $4$ term, and $R_x$ can be
replaced with $R_y$ at the cost of a rank $3$ term.  This leads to 
\eqa{k-rnc1}{
K &\sim 
\bint H\exp\Bigg[ i\bracket{\rho, \psi_x-
  \psi_y} 
-\frac{t\scalar }{6} + \frac{1}{12}
\Ricci\parens{\vec{x}_y, \vec{x}_y} + \frac{i}{6}\bracket{\rho, R\sbrace{\vec{x}_y, \psi_x}
  \vec{x}_y} \nonumber\\
&\qquad + 
\frac{t}{8} \parens{\rho, R\sbrace{\psi_x, \psi_x} \rho} + 
\frac{t}{8} \parens{\rho, R\sbrace{\psi_y, \psi_y} \rho}   \Bigg]
d\rho.
}
That is, defining $K'$ as
\eqa{k-rnc}{
K' &\equiv \bint H \exp\Bigg[i\bracket{\rho, \psi_x-
  \psi_y} 
-\frac{t\scalar }{6} + \frac{1}{12}
\Ricci\parens{\vec{x}_y, \vec{x}_y} + \frac{i}{6}\bracket{\rho, R\sbrace{\vec{x}_y, \psi_x}
  \vec{x}_y} \nonumber\\
&\qquad + 
\frac{t}{8} \parens{\rho, R\sbrace{\psi_x, \psi_x} \rho} + 
\frac{t}{8} \parens{\rho, R\sbrace{\psi_y, \psi_y} \rho}   \Bigg]
d\rho,
}
\[
K \sim K'.
\]

             \subsection{The Heat Operator}

Let $\Delta_{\text{LdR}}$ be the operator on forms given by 
\be
\Delta_{\text{LdR}}=-g^{ij} \parens{\Nabla_i \Nabla_j - \Gamma_{ij}^k \Nabla_k}
+ \frac{1}{2}\parens{\frac{\partial}{\partial \psi}, R\sbrace{\psi, \psi}
  \frac{\partial}{\partial \psi}}
\ee{Delta-def}
where $\Delta_i=\partial_i - \Gamma_{ij}^k \psi^j \partial/\partial \psi^k.$

\begin{proposition} \label{pr:heat-est}
For $K=K(x,y,\psi_x, \psi_y;t)$ defined by Eq.~\eqref{eq:K-def} and
$\Delta$ by Eq.~\eqref{eq:Delta-def}
\eqa{heat-est}{
\frac{\partial K}{\partial t}  + \frac{1}{2} \Delta_x K &=
F_1 H\parens{x,y;t}+  F_2 t \mbox{, and} \nonumber \\
\frac{\partial K}{\partial t} + \frac{1}{2} \Delta_y K &= 
F_3 H\parens{x,y;t}+  F_4 t,
}
where $F_i \in E_{x,y}$ and $\abs{F_i}_t= \OO\parens{t^\epsilon}.$
\end{proposition}

\begin{proof}

  Recalling Eq.~\eqref{eq:rank-bd} of Lemma~\ref{lm:tnorm-facts}, the
  argument reduces to showing that the heat operators on the left-hand
  sides applied to $K$ give terms of rank $1$ times $H.$ To this end,
  first observe that the heat operator acts on $H$
  times terms of a given rank to produce $H$ times terms whose rank is
  lower by at most $2$. This is because the operator
  $\partial/\partial t$ lowers the rank of a smooth function of $t$ by
  $2,$ and, when applied to $H$, gives $H\delta_{ij} x^j x^j/2t^2,$ which
  also lowers the rank of whatever it multiplies by $2.$ In
  $\Delta_x$, the curvature term does not change rank at all (and is
  zero on $H$).  The operator $\Nabla_i$ in the other hand, lowers
  rank by one for smooth terms because $\partial_i$ lowers the degree
  of $\vec{x}_y$ by $1$ and the $\Gamma_{ij}^k \psi^j \partial/
    \partial \psi^k$ actually increases rank by $1$ (by
  Eq. ~\eqref{eq:gamma-rnest}) Acting on $H$, $\Nabla_i$ gives
  $-H \delta_{ij} x^j/t,$ which lowers the rank of whatever it multiplies
  by $1.$ Thus, taken together, the heat operator will lower rank by
  at most two. As a first consequence, if the  heat operator applied to some $\hat{K} \sim K$
  gives terms of rank $1$ times $H$, then the same is true for $K$
  itself. This in turn implies the two equations in the
  proposition are equivalent, since, as noted above, $K$ is symmetric
  up to terms of rank $3$ times $H$.

It suffices, then, to confirm the result on $K' \sim K,$ where $K'$ is
defined in Riemann normal coordinates by the right-hand side of
Eq.~\eqref{eq:k-rnc}.  In fact, simplifying the heat operator itself
in Riemann normal coordinates, by dropping terms that when applied to $K'$
give $H$ times terms of rank $1$ or higher, 
will lead to an operator $\Delta_2$ and
a kernel $K_2 \sim K'$ for which, by
Eq.~\eqref{eq:heat2-est} of Lemma~\ref{lm:approx-heat}, the
proposition holds.

Begin by expanding $\Delta_{\text{LdR}}$ in Riemann normal coordinates,
\begin{align*}
\Delta_{\text{LdR}}&=-g^{ij} \partial_i \partial_j + 2 g^{il}
\Gamma_{iq}^p \psi^q \frac{\partial\,}{\partial \psi^p} \partial_l +
g^{il} \partial_l \parens{\Gamma_{iq}^p} \psi^q \frac{\partial\,}{\partial \psi^p}
-g^{il} \Gamma_{iq}^p \Gamma_{lk}^n \psi^q \frac{\partial\,}{\partial
  \psi^p}\psi^k \frac{\partial\,}{\partial \psi^n}\\
&\qquad +g^{il}\Gamma_{il}^q \partial_q - g^{il}\Gamma_{il}^q\Gamma_{qp}^k \psi^p \frac{\partial\,}{\partial
  \psi^k} + \frac{1}{2}\parens{\frac{\partial}{\partial \psi}, R\sbrace{\psi, \psi}
  \frac{\partial}{\partial \psi}}\\
&\sim- 2\Delta_0 + 2 
\delta^{il}\Gamma_{iq}^p \psi^q \frac{\partial\,}{\partial \psi^p} \partial_l +
\delta^{il}\partial_l \parens{\Gamma_{iq}^p} \psi^q \frac{\partial\,}{\partial \psi^p} + \frac{1}{2}\parens{\frac{\partial}{\partial \psi}, R\sbrace{\psi, \psi}
  \frac{\partial}{\partial \psi}}\\
&\sim - 2\Delta_0 - \frac{2}{3}\delta^{il}\parens{
R_{ikq}^{\,\,\,\,\,\,p} + R_{qki}^{\,\,\,\,\,\,p}}x^k \psi^q
\frac{\partial\,}{\partial \psi^p} \partial_l +\frac{1}{3}\delta^{ik} R_{qki}^{\,\,\,\,\,\,p}\psi^q \frac{\partial\,}{\partial \psi^p} + \frac{1}{2}\parens{\frac{\partial}{\partial \psi}, R\sbrace{\psi, \psi}
  \frac{\partial}{\partial \psi}}
\end{align*}
where $\Delta_0$ is defined in Lemma~\ref{lm:approx-heat}.  
 The  terms
omitted, when applied to $K'$, 
give terms of rank $1$ or higher times $H$
 (each
$\partial_l$ lowers rank by $1,$ each factor of $\Gamma$ increases
rank by
$1$). 

Thus
\[\Delta_{\text{LdR}}= -2 \sbrace{ \Delta_0 + h_k^l x^k \partial_l +
  h_k^k/2 + f\parens{\vec{x}} }\]
 where $h_k^l= \frac{1}{3}\delta^{il}\parens{
R_{ikq}^{\,\,\,\,\,\,p} + R_{qki}^{\,\,\,\,\,\,p}}\psi^q \frac{\partial\,}{\partial \psi^p}$
and $f= - \frac{1}{4}\parens{\frac{\partial}{\partial \psi}, R\sbrace{\psi, \psi}
  \frac{\partial}{\partial \psi}}.$  The expression in Riemann
normal coordinates denotes a differential operator on $\RR^n$ with
coefficients in $\End_0\parens{\Lambda^* \RR^n},$ so it still makes
sense to refer to the degree and rank of a term. 
Of course, this is
a noncommutative algebra, but since $e^{-h_m^k x^m x^k/2}$ and
$e^{-tf}$ commute up to terms of rank $4,$
Eq.~\eqref{eq:heat2-est} applies, with $C$ of degree $2$ and $D$ of
degree $1,$ to conclude 
\[\frac{\partial K_2}{\partial t} + \frac{1}{2} \Delta_\text{LdR} K_2=
FK_2\]
where
\[K_2\parens{\vec{x};t}= H e^{  \frac{1}{12} \Ricci_{ij} x^i x^j + \frac{t \scalar}{12}
  - \frac{1}{6} 
R_{kij}^{\,\,\,\,\,\,l}x^i x^j\psi^k \frac{\partial\,}{\partial \psi^l} 
-\frac{t}{4} \parens{\frac{\partial\,}{\partial \psi}, \overline{R}\sbrace{\psi, \psi}
  \frac{\partial\,}{\partial \psi}}},
\] 
$\overline{R}=\frac{1}{2}\parens{R_x + R_y},$ and $F,$ representing all the contributions to the error term in
Eq.~\eqref{eq:heat2-est}, is of rank $1.$ 

 It
remains to show  $K_2 \sim K'$, when the first is expressed as a
kernel integrated over 
$\rho$ and the second is given by  Eq.~\eqref{eq:k-rnc}.

Towards that goal note 
\begin{align*}
 &2\parens{\frac{\partial\,}{\partial \psi}, \overline{R}\sbrace{\psi, \psi}
  \frac{\partial \,}{\partial \psi}}= R_{ijkl}(x) \frac{\partial
  \,}{\partial \psi^l} \psi^i \psi^j \frac{\partial \,}{\partial
  \psi^k} + R_{ijkl}(y) \frac{\partial
  \,}{\partial \psi^l} \psi^i \psi^j \frac{\partial \,}{\partial
  \psi^k}\\
&\qquad = - R_{ijkl}(x)
 \psi^i \psi^j \frac{\partial \,}{\partial
  \psi^k}  \frac{\partial \,}{\partial
  \psi^l} + R_{ijki}(x) \psi^j \frac{\partial \,}{\partial
  \psi^k}  -R_{ijkj}(x) \psi^i \frac{\partial \,}{\partial
  \psi^k} \\
&\qquad \qquad - R_{ijkl}(y) \frac{\partial \,}{\partial
  \psi^k}\frac{\partial
  \,}{\partial \psi^l} \psi^i \psi^j +R_{ijjl}(y) \frac{\partial
  \,}{\partial \psi^l} \psi^i - R_{ijil}(y)\frac{\partial
  \,}{\partial \psi^l} \psi^j\\
&\qquad = - R_{ijkl}(x)
 \psi^i \psi^j \frac{\partial \,}{\partial
  \psi^k}  \frac{\partial \,}{\partial
  \psi^l} - R_{ijkl}(y) \frac{\partial \,}{\partial
  \psi^k}\frac{\partial
  \,}{\partial \psi^l} \psi^i \psi^j + 2 R_{ijki}(x) \psi^j \frac{\partial \,}{\partial
  \psi^k} +2 R_{jiil}(y) \frac{\partial
  \,}{\partial \psi^l} \psi^j\\
&\qquad = - R_{ijkl}(x)
  \psi^i \psi^j  \frac{\partial \,}{\partial
  \psi^l}\frac{\partial \,}{\partial \psi^k} - R_{ijkl}(y) \frac{\partial \,}{\partial
  \psi^k}\frac{\partial
  \,}{\partial \psi^l} \psi^i \psi^j + 2 \scalar(y) +
\text{Rank 1}.
\end{align*}
Multiplying by $-t/8$, exponentiating and writing the resulting
operator as a kernel integrated over $\rho$ gives
\[ H e^{-\frac{t}{4} \parens{\frac{\partial\,}{\partial \psi}, \overline{R}\sbrace{\psi, \psi}
  \frac{\partial\,}{\partial \psi}}} \sim \bint H e^{i \bracket{\rho,
  \psi_x - \psi_y} -\frac{t\scalar}{4} + \frac{t}{8}\sbrace{\parens{\rho, R\sbrace{\psi_x, \psi_x} \rho}
+ \parens{\rho, R\sbrace{\psi_y, \psi_y} \rho}}} d\rho.\]
The rank 1 error term disappears,
since the exponential of $t$ times a rank 1 term is $1$ plus a rank $3$ term.
Combining this with  the other Grassman term from $K_2$, gives, as an
integral over $\rho$,
\begin{align*}
&He^{- \frac{1}{6} 
R_{kmj}^{\,\,\,\,\,\,\,\,\,l}x^m x^j\psi^k \frac{\partial\,}{\partial \psi^l}-\frac{t}{4} \parens{\frac{\partial\,}{\partial \psi}, \overline{R}\sbrace{\psi, \psi}
  \frac{\partial\,}{\partial \psi}}} \\
 & = \bint H e^{i \bracket{\rho,
  \psi_x - \psi_y} -\frac{t\scalar}{4} - \frac{i}{6} 
R_{kmj}^{\,\,\,\,\,\,\,\,\,l}x^m x^j\rho_l\psi_x^k + \frac{t}{8}\sbrace{\parens{\rho,
    R\sbrace{\psi_x, \psi_x} \rho} 
+ \parens{\rho, R\sbrace{\psi_y, \psi_y} \rho}}} d\rho.
\end{align*}
Comparing the definition of $K_2$ above with that of $K'$ in Eq.~\eqref{eq:k-rnc}
 gives $K_2 \sim K'$.
\end{proof}


\section{The $*$ Product of Kernels and the Norm}

The material in this section relies heavily on \S\ref{ss:bundle}
and~\S\ref{ss:norm}, requires
the definition Eq.~\eqref{eq:K-def} of $K,$ and makes one brief
reference to Eq.~\eqref{eq:dist-est}.  Lemmas~\ref{lm:h-bd}
and~\ref{lm:k-bd} involve rather technical bounds, and are only
necessary for the proofs of
Propositions~\ref{pr:norm} and~\ref{pr:kk-est}.  Only
Eqs.~\eqref{eq:k*l}, \eqref{eq:l*l}, and~\eqref{eq:kk-est} are needed
going forward for the rest of the paper.  Notice the norm $\norm{\, \cdot
 \,}_t$ could not be defined until Lemmas~\ref{lm:h-bd}
and~\ref{lm:k-bd} guarantee the existence the quantity $D$ in the
definition so that Eq.~\eqref{eq:k*l} holds.  Once the norm is defined,
Corollary~\ref{cr:heat-est2} can be viewed as the correct statement of
Proposition~\ref{pr:heat-est}.

            \subsection{Bounds on $K$ and $H$}

Recall that $L_1, L_2 \in E_{x,y}$ represent operators on
$\Gamma(\Lambda^* TM)$ and the product of operators corresponds to the
$*$ product 
\[
L_1*L_2(x,z,\psi_x, \psi_z) = \int \bint L_1(x,y,\psi_x,
\psi_y)L_2(y,z,\psi_y,\psi_z) d\psi_y\, dy.
\]

\begin{lemma} \label{lm:h-bd}
For bounded $t>0,$ and $H$ as in Eq.~\eqref{eq:h-def}, 
\be
\int H(x,y;t)dy= 1+ \OO(t);
\ee{hint-bd}
for bounded $t_1,t_2>0$ with $t=t_1+t_2$, and  $\epsilon$ as
in \S\ref{ss:norm},
\eqa{hh-bd}{
H(t_1)*H(t_2) = e^{\OO\parens{t_1^{\epsilon/n}}}\sbrace{H(t)+\OO\parens{t^2}},
\nonumber\\
H(t_1)*H(t_2) =
e^{\OO\parens{t_2^{\epsilon/n}}}\sbrace{H(t)+\OO\parens{t^2}}.
}
\end{lemma}

\begin{proof}
By Eq.~\eqref{eq:g-rnest} the measure
$dx=\sbrace{1+\OO\parens{\abs{\vec{x}_y}^2}}d\vec{x}_y$ where
$d\vec{x}_y$ comes from the inner product on $T_yM,$ so
Eq.~\eqref{eq:hint-bd} becomes
\[\int H\parens{x,y;t} dx = \int (2\pi t)^{-n/2}
e^{-\frac{\abs{\vec{x}_y}^2}{2t}}
\sbrace{1+\OO\parens{\abs{\vec{x}_y}^2}} d\vec{x}_y =
1+\OO\parens{t}.\]

For Eq.~\eqref{eq:hh-bd}, following the $t^{1/2-\delta}$-trick in the
proof of Lemma~\ref{lm:tnorm-facts},  consider first the 
case where
$\abs{\vec{z}_x}$ is greater  than some
multiple of $t^{1/2-\delta}.$ 
Then either $\abs{\vec{x}_y}$ is more than a multiple of
$t_1^{1/2-\delta}$  or $\abs{\vec{y}_z}$ is more than a
multiple of $t_2^{1/2-\delta}$ so that the entire integral 
\[
H(t_1)*H(t_2) (x,z)  = (2\pi t_1)^{-n/2} (2\pi t_2)^{-n/2} \int
e^{-\frac{\abs{\vec{x}_y}^2}{2t_1} -\frac{\abs{\vec{y}_z}^2}{2t_2}}
dy\]
 is bounded
by a multiple of $t^2$ times the integral of a multiple of
$H(x,y;t),$ which by Eq.~\eqref{eq:hint-bd} is bounded by a multiple of
$t^2.$

If $\abs{\vec{z}_x}$ is less than a multiple of $t^{1/2-\delta},$
choose the multiple small enough first of all that $x$ and $z$ are within the
injectivity radius of each other.  Define 
\[E(x,y,z;\alpha)=(1-\alpha)\abs{\vec{x}_y}^2 +
\alpha\abs{\vec{z}_y}^2 - 
\alpha (1-\alpha) \abs{\vec{z}_x}^2\]
on the compact domain where $\alpha \in [0,1]$ and the three lengths
are all well-defined, so that if $\alpha=t_1/t$  then
\[
\frac{\abs{\vec{x}_y}^2}{t_1} +
\frac{\abs{\vec{y}_z}^2}{t_2} = \frac{\abs{\vec{z}_x}^2}{t} +
\frac{t}{t_1t_2} E.
\]
Let $u$ be the point on the geodesic between $x$ and $z$ that is
$\alpha$ of the way from $x$ to $z,$ so that $\vec{u}_x=
\alpha\vec{z}_x$ and $\vec{u}_z= (1-\alpha) \vec{x}_z.$ 
Apply Eq.~\eqref{eq:dist-est}, to rewrite $\abs{\vec{x}_y}^2, \abs{\vec{y}_x}^2$ and $\abs{\vec{z}_y}^2$
in terms of $\abs{\vec{y}_u}$, $\abs{\vec{u}_x}$ and $\abs{\vec{u}_z}$, and then express
these in terms of $\abs{\vec{z}_x}$ to obtain
$E(x,y,z;\alpha)=\abs{\vec{y}_u}^2\sbrace{1 + \OO\parens{\alpha(1-\alpha)\abs{\vec{z}_x}^2}}.$

Choosing $\delta$ and the
coefficient of $t^{1/2-\delta}$
wisely,  $E$ is bounded 
 below by $\abs{\vec{y}_u}^2\sbrace{1-C\alpha(1-\alpha)
   t^{\epsilon}}$ where
 $C$ is small enough so this is always positive.  Also bound $dy$ by
 $\parens{1+D\abs{\vec{y}_u}^2}d\vec{y}_u$ for some $D>0,$ so
\begin{align*}
& H(t_1)*H(t_2) (x,z)
=
\int (2\pi t_1)^{-n/2} (2\pi t_2)^{-n/2}
e^{-\frac{\abs{\vec{x}_y}^2}{2t_1} - \frac{\abs{\vec{z}_y}^2}{2t_2}}
dy \\
& \quad \leq (2 \pi t)^{-n/2} e^{-\frac{\abs{\vec{z}_x}^2}{2t}}\int ( 2
 \pi t_1 t_2/t)^{-n/2}
 e^{-\frac{t\abs{\vec{y}_u}^2}{2t_1t_2}\sbrace{1-C\alpha(1-\alpha)t^\epsilon}}
 \sbrace{1+D\abs{\vec{y}_u}^2}d\vec{y}_u\\
&\quad = H(x,z;t)
 \sbrace{1-C\alpha(1-\alpha)t^\epsilon}^{-n/2} \parens{1+D\alpha
   (1-\alpha)t^\epsilon}^{n/2}.
\end{align*}
From here each estimate in Eq.~\ref{eq:hh-bd} is immediate. In fact,
this shows the bound could be strengthened, replacing the exponential
factor by $1 + \OO(t^\epsilon)$, but the weaker form is sufficient and
more convenient later.
\end{proof}

\begin{lemma} \label{lm:k-bd}
For $t > t_1 > 0$ bounded, $K(t_1)=K(x,y,\psi_x, \psi_y;t_1)$ in
Eq.~\eqref{eq:K-def}, $\epsilon$  and $\abs{\,\cdot\,}_t$  as
in \S\ref{ss:norm}, 
and $\overline{F} \in \overline{E}_{x,y},$
\eqa{ka-bd}{
\abs{\overline{K}(x,y;t_1)\overline{F}}_t&= e^{\OO\parens{t_1^{\epsilon/n}} }\abs{\overline{F}}_t
 \sbrace{H(x,y;t_1) + \OO(t_1)}\nonumber \\
\abs{\overline{F}\overline{K}(x,y;t_1)}_t&= e^{\OO\parens{t_1^{\epsilon/n}} }\abs{\overline{F}}_t
 \sbrace{H(x,y;t_1) + \OO(t_1)}
}
and if 
$f(x,\psi_x) \in \Gamma(\Lambda^* TM).$
\be
\lim_{t \to 0} K(t) * f = f.
\ee{tlim}
\end{lemma}
\begin{proof}
For Eq.~\eqref{eq:ka-bd}, notice that $\overline{K}$ can be written
as $\parens{1+\overline{L}}H,$ where $\overline{L}$ is a sum of
products of terms which are either rank $2$ and degree at most $1$ (namely, $e^{ \frac{1}{12} \Ricci_{kl}
 x^k x^l + \frac{t_1 \scalar}{12}}-1,$ or $- \frac{1}{6} 
R_{kmj}^{\,\,\,\,\,\,\,\,\,l}x^m x^j\psi^k \frac{\partial\,}{\partial \psi^l}$) or a
multiple of $t_1$ and degree at most $2$ (namely, $-\frac{t}{4} \parens{\frac{\partial\,}{\partial \psi}, \overline{R}\sbrace{\psi, \psi}
 \frac{\partial\,}{\partial \psi}}$).  The result follows from
Eq.~\eqref{eq:almost-contract}.

For Eq.~\eqref{eq:tlim}
let $R$ be the the set of $x$ for which $K(x,y;t)$ is nonzero for a
given $y.$ Notice $K(x,y,\psi_x, \psi_y;t) = He^{ i \bracket{\rho, \psi_x - \psi_y}}\sbrace{1 +
\OO\parens{\abs{\vec{y}_x}+t}}$, and, if $f$ is smooth, $f(y, \psi_y)
= f(x, \psi_y) +  \OO\parens{\abs{\vec{y}_x}}$, so
\begin{align*}
\sbrace{K(t)*f}(x,\psi_x)&= \int_R \bint K(x,y,\psi_x, \psi_y;t)f(y,
\psi_y) d\psi_y\, dy\\
&= \int_R \bint \bint (2\pi t)^{-n/2} e^{-\frac{\abs{\vec{y}_x}^2}{2t}
 + i \bracket{\rho, \psi_x - \psi_y}} \sbrace{f(x, \psi_y) +
 \OO\parens{\abs{\vec{y}_x}+t}} d\rho \, d\psi_y\, d\vec{y}_x\\
&= \int_R (2\pi t)^{-n/2} e^{-\frac{\abs{\vec{y}_x}^2}{2t}
 } \sbrace{f(x, \psi_x) +
 \OO\parens{\abs{\vec{y}_x}+t}} d\vec{y}_x\\
&= f(x,\psi_x)+ \OO\parens{t^{1/2}}.
\end{align*}
\end{proof}

\subsection{The Norm on Kernels}

Error terms arising in the following sections
will typically be kernels of the form
\[\sbrace{F_1\parens{x,y;t} H\parens{x,y;t} + F_2\parens{x,y;t}}
P\parens{x,y}\]
where $F_1,F_2\in E_{x,y}.$  The size of $\abs{F_i}_t$ gives a rough
measure of the size of these errors, but more precision requires additional details on
the size of $F_2.$  
\begin{proposition} \label{pr:norm}
Given $\epsilon$ as in \S\ref{ss:norm}, there exists a $D>0$ such
that, defining  $\norm{L}_t$ for $L \in E_{x,y}$ to be the least number such that 
\be
\abs{\overline{L}}_t\leq \norm{L}_t\sbrace{H(t) + Dt},
\ee{norm-def} 
this norm satisfies, for $K$ defined in Eq.~\eqref{eq:K-def} and for bounded and positive $t_1$ and $t_2$ with
$t_1+t_2=t$,
\eqa{k*l}{
\norm{K(t_1)*L}_t &\leq e^{\OO\parens{t_1^{\epsilon/n}} + \OO\parens{t^2}}\norm{L}_{t_2} \nonumber\\
\norm{L*K(t_2)}_t &\leq e^{\OO\parens{t_2^{\epsilon/n}} + \OO\parens{t^2}}\norm{L}_{t_1}
}
and
\be
\norm{L_1*L_2}_t = \OO\parens{t^{-1}} \norm{L_1}_{t_1} \norm{L_2}_{t_2}
\ee{l*l}
\end{proposition}
\begin{proof}
Observe that 
\[ \overline{L_1*L_2}(x,z)= \int \overline{L}_1(x,y) \ptxystar
\overline{L_2}(y,z) \ptyxstar dy,\]
 and that $\abs{\,\cdot\,}_t$ is unchanged by parallel transport.  
For the first line of Eq.~\eqref{eq:k*l}, apply this in the special
case $L_1 = K(t_1)$. Then, using Eq.~\eqref{eq:ka-bd} to define
$b_1,D>0$, applying property (c) of Lemma~\ref{lm:tnorm-facts} with $t_2
< t$, using Eq.~\eqref{eq:hh-bd} to define
 $b_2,D_2>0$, and using Eq.~\eqref{eq:hint-bd} to define
$D_3$,
\begin{align*}
\abs{\overline{K(t_1)*L}}_t&\leq \int \abs{\overline{K}(x,y;t_1)\ptxystar\overline{L}(y,z)\ptyxstar}_tdy\\
&\leq  e^{b_1t_1^{\epsilon/n}} \int \abs{\overline{L}}_t 
\sbrace{H(x,y;t_1)+D t_1}dy\\
&\leq  e^{b_1t_1^{\epsilon/n}} \norm{L}_{t_2}\int \sbrace{H(y,z;t_2) + Dt_2}
\sbrace{H(x,y;t_1)+D t_1}dy\\
&\leq  e^{b_1t_1^{\epsilon/n}} \norm{L}_{t_2} \Big[e^{b_2t_1^{\epsilon/n}}
 H(x,z;t) + Dt_2 + D{t_1} +D_2t^2  \\
& \qquad \qquad  \qquad  \qquad
 + D D_3 t_1 t_2 + D D_3 t_1 t_2 + D^2 t_1t_2 \vol(M) \Big]\\
& \leq e^{(b_1 + b_2)t_1^{\epsilon/n}} \norm{L}_{t_2}\sbrace{ H(x,z;t) +
Dt + D_2 t^2 + \parens{2D D_3 + D^2 \vol(M)} t_1t_2}\\
&\leq  e^{bt_1^{\epsilon/n}+ ct^2} \norm{L}_{t_2} \sbrace{
 H(x,z;t) + Dt}.
\end{align*}
Here choose  $b \geq b_1+b_2$ and
$c\geq D_2 + 2DD_3 + D^2 \vol(M).$   
 The second version of Eq.~\eqref{eq:k*l} works
in exactly the same fashion.

For Eq.~\eqref{eq:l*l}, for some $C_1>0$ coming from
Eq.~\eqref{eq:norm-prod}, $C_2>0$ coming from Eq.~\eqref{eq:hh-bd} and
$C_3>0$ coming from Eq.~\eqref{eq:hint-bd} and~\eqref{eq:hh-bd} (and depending on $D$)
\begin{align*}
\abs{\overline{L_1*L_2}}_t & \leq \int \abs{\overline{L_1}(x,y)\ptxystar
\overline{L_2}(y,z) \ptyxstar}_tdy\\
& \leq C_1 t^{-1}\int \abs{\overline{L_1}(x,y)}_{t_1}
\abs{\overline{L_2}(y,z)}_{t_2} dy\\
& \leq C_1\norm{L_1}_{t_1} \norm{L_2}_{t_2} t^{-1}\int
\sbrace{H(x,y;t_1) + Dt_1} \sbrace{H(y,z;t_2) + Dt_2}dy\\
& \leq C_1\norm{L_1}_{t_1} \norm{L_2}_{t_2} t^{-1}\sbrace{C_2H(x,z;t)
 + C_3 D t}\\
& \leq C_4\norm{L_1}_{t_1} \norm{L_2}_{t_2} t^{-1}\sbrace{H(x,z;t)
 + D t}
\end{align*}
for $C_4>C_1 \max(C_2,C_3).$
\end{proof}
With this definition, Eq.~\ref{eq:heat-est} of
Proposition~\ref{pr:heat-est} immediately implies
\begin{corollary}
If $\epsilon$ is as in \S\ref{ss:norm}, $K$ as in
Eq.~\eqref{eq:K-def}  and $\Delta$ as in Eq.~\eqref{eq:Delta-def}
\eqa{heat-est2}{
\frac{\partial}{\partial t} K &= -\frac{1}{2} \Delta_x K +
\OO_t\parens{t^\epsilon} \nonumber \\
&= -\frac{1}{2} \Delta_y K +
\OO_t\parens{t^\epsilon} 
}
where $\OO_t\sbrace{f(t)}$ refers to an arbitrary kernel bounded in
$\norm{\,\cdot\,}_t$ by a multiple of $f(t).$ \label{cr:heat-est2}
\end{corollary}

\subsection{Approximate Semigroup Property}

\begin{proposition} \label{pr:kk-est}
For $K(t)=K(x,y,\psi_x, \psi_y;t)$ in Eq.~\eqref{eq:K-def}, the norm
$\norm{\,\cdot\,}_t$ as in Proposition~\ref{pr:norm} and
$t=t_1+t_2>0$ 
\be
\norm{K(t_1) * K(t_2) - K(t)}_t = \OO\parens{t^{1+\epsilon}}
\ee{kk-est}
\end{proposition}

\begin{proof}

Let $r$ be the injectivity radius of $M,$ so that $K(x,y;t)=0$ if $x$
and $y$ are more than $r$ apart.  Suppose first that the spatial
variables $x$ and $z$ implicit on the left-hand side of Eq.~\eqref{eq:kk-est} are
more than $r/2$ apart.  By the triangle inequality, 
$\abs{\overline{K(t_1)*K(t_2)} - \overline{K}(t) }_t \leq
\abs{\overline{K(t_1)*K(t_2)}}_t + \abs{\overline{K}(t)}_t$. The
second summand is bounded by a constant times $H(x,z; t)$, which, because
the distance between $x$ and $z$ is bounded below, is exponentially
damped for small $t$; in particular, it is $\OO(t^{1+\epsilon})$. The first summand is bounded by
some non-negative power of $t$ times $H(t_1)* H(t_2).$ 
Since within the integral represented by the $*$ product every point
$y$ in the domain is a distance at least $r/4$ from either $x$ or $z,$
this integral is likewise exponentially damped, which proves Eq.~\eqref{eq:kk-est} in this case.

If 
$x$ and $z$ are less than $r/2$ apart,
\begin{align*}
&K(t_1)* K(t_2)- K(t)  =  \int_{0}^{t_1} \partial\sbrace{K(\tau)* K(t-\tau)}/\partial \tau\, d\tau\\
&= \int_{0}^{t_1} \dot{K}(\tau)* K(t-\tau)- K(\tau) * \dot{K}(t-\tau) \, d\tau\\
&= \int_{0}^{t_1} \OO_\tau\parens{\tau^\epsilon}*K(t-\tau) +  K(\tau)* \OO_{t-\tau}\parens{\sbrace{t-\tau}^\epsilon}\, d\tau \\
&\qquad + \int_{0}^{t_1} - \Delta_y\sbrace{K(\tau)} *K(t-\tau) d\tau + K(\tau)* \Delta_y\sbrace{K(t-\tau)} \,  d\tau\\
&= \int_{0}^{t_1}  \OO_t\parens{\tau^\epsilon} + \OO_t\sbrace{(t-\tau)^\epsilon} \, d\tau \\
&\qquad + 
\int_{0}^{t_1}\sbrace{\int_{\partial_y} \star d_y \sbrace{K(\tau)}
  K(t-\tau) dy
 - K(\tau) \star d_y\sbrace{K(t-\tau)} dy}\,  d\tau
\end{align*}
where the third equation comes from Eq.~\eqref{eq:heat-est2}, the
fourth from Eq.~\eqref{eq:k*l}, and
boundary term from taking the formal adjoint of $\Delta_{\text{LdR}}.$
The  $\int_{\partial _y}$ indicates the integral in the $y$ variable
of the given forms along the boundary of the earlier region of
integration, and $\star$ denotes the Hodge star operator.

The first integral reduces to $\OO_t(t^{1+\epsilon})$. For the boundary integral, the domain of integration
has bounded metric volume (bounded by the sum of the volumes of the
spheres of the injectivity radius around $x$ and $z$), so to bound the
integral it suffices to bound the integrand.  This integrand is of the
form a polynomial in $\tau,t-\tau,\tau^{-1},$ and $(t-\tau)^{-1}$ times a smooth
function of $x,y,z$ times $e^{-\abs{\vec{y}_x}^2/2\tau
  -\abs{\vec{y}_z}^2/2(t-\tau)}.$   Each point on the boundary is a
distance $r$ from either $x$ or $z,$ and since $x$ and $z$ are less
than $r/2$ apart, each point is a distance at least $r/2$ from both.
Thus the Gaussian is exponentially damped in $\tau$ and $t-\tau,$ so the
whole expression is bounded by a
multiple of any power of 
  $t_1$ and $t_2.$  In particular it is $\OO\parens{
    t^{1+\epsilon}}=\OO_t\parens{t^\epsilon}.$ Eq.~\eqref{eq:kk-est} then follows upon
  completing the now trivial $\tau$ integral.
\end{proof}

\section{Partitions}

\subsection{Partitions and the Pointwise Limit}

For $t>0,$ a \emph{partition} $P$ of $t$ is a sequence $P=(t_1, t_2,
\ldots, t_m)$ with $t= \sum_i t_i.$  If $P$ is a  partition of $t$ and
$P'$ is a partition of $t',$ then the concatenation $PP'$ is a
partition of $t+t';$  if $P_i$ is a partition of $t_i$ for $1 \leq
i \leq m,$ then $P_1P_2 \cdots P_m$ is a \emph{refinement} of $P=(t_1,
\ldots, t_m).$  Define $\abs{P}=\max t_i.$  If
$P$ is a partition of $t$, and $K$ is any kernel,  define
\be
K^{*P}= K(t_1)*K(t_2) * \cdots * K(t_m)
\ee{k*p-def} 
Take $\{P :\abs{P}<r\}$ for $r > 0$ to define open sets in   the space of all partitions of a given $t$. Then it makes sense to define $\lim_{\abs{P} \to 0}$ of a
function depending on $t$ and $P$.

\begin{proposition} \label{pr:op-limit}
For $K$ as defined in Eq.~\ref{eq:K-def},  $f_0(x, \psi_x)$ piecewise continuous, and $t \geq 0$, 
\be
f(t) = \lim_{\abs{P} \to 0} K^{*P} * f_{0}
\ee{op-limit}
is the unique solution  $f(x,\psi_x;t)$ to the heat equation $\partial f/\partial t =-\frac{1}{2} \Delta_\text{LdR}
f$ with  $f(0) = f_0$.
\end{proposition}

\begin{proof}
Let $f(t)$ be the solution to the heat equation with $f(0)=f_0$ (the
dependence on $x$ and $\psi_x$ is understood).  
The definition of $K$ ensures that $K(t)*f(s)$ is a smooth
function of $t$ 
 with $\lim_{t \to 0} K(t)*f(s)=f(s).$  Then
$f(t_1+t_2)-f(t_2)= t_1 \Delta_{\text{LdR}} f(t_2) +
  \OO\parens{t_1^2}$ (all explicit and implicit constants in this
  proof depend on $f$ and thus $f_0$).  Using Eq.~\ref{eq:tlim} to expand about
  $t_1 = 0$, and~\eqref{eq:heat-est2} 
\[K(t_1) * f(t_2)- f(t_2)  = t_1 \l.\frac{\partial K}{\partial
  t}\r|_{t = 0} * f(t_2) +
\OO\parens{t_1^2}= t_1 \Delta_{\text{LdR}} f(t_2) +
  \OO\parens{t_1^{1+\epsilon}}.\]
This uses the fact that for $E(t) = \OO_t(t^\epsilon)$, $E(t) * f
= \OO(t^\epsilon).$

Thus there is a constant $C$ such that
$\norm{K(t_1)*f(t_2)- f(t_1+t_2)}_\infty\leq C t_1^{1+\epsilon}$ for all
$t_1+ t_2 \leq t.$   Since the operator $K*$ depends smoothly on
$t$ and approaches the identity as $t$ goes to $0$ we can also select $C$ so that the
operator norm of $K(t_1)$ is less than $e^{C t_1}.$ 

Therefore decomposing $P$  for each $t_k
\in P$ as $P_1(t_k)P_2$, and writing $\sbrace{K^{*P} * f(0)-
f(t)}$ as an appropriate telescoping sum, 
\begin{align*}\norm{\sbrace{K}^{*P} * f(0)- f(t)}_\infty&\leq \sum_k
  \norm{\sbrace{K}^{*P_1} * \sbrace{K(t_k)*f\parens{\sum_{i > k} t_i}- f\parens{t_k
      + \sum_{i > k} t_i} }}_\infty \\ 
&\leq \sum_k e^{C\sum_{i < k} t_i} C t_k^{1+\epsilon} \leq Ce^{Ct} t \abs{P}^\epsilon
\end{align*}
which converges to $0$ under refinement.


\end{proof}

\subsection{The Cauchy Property}
\begin{proposition} \label{pr:cauchy}
For sufficiently small $t>0,$ 
all partitions $P$ of
$t$,  all refinements $Q$ of $P$, and with the norm defined in Prop.~\ref{pr:norm}
\be
\norm{K^{*Q} -K^{*P}}_t =\OO(t)\abs{P}^\epsilon.
\ee{k*k-refine}
\end{proposition}

\begin{proof}
Eqs.~\eqref{eq:k*l}, \eqref{eq:l*l}, and~\eqref{eq:kk-est} give
constants $a,b,c,d$ such that for all $t=t_1+t_2$ positive and bounded
and all $L, L_1, L_2 \in E_{x,y}$
\eqa{abc}{
\norm{K(t_1)*K(t_2) - K(t)}_t &\leq at^{1+\epsilon} \nonumber \\
\norm{K(t_1)*L}_t&\leq e^{bt_1^{\epsilon/n}+ct^2}\norm{L}_{t_2}\nonumber \\ 
\norm{L*K(t_2)}_t&\leq e^{bt_2^{\epsilon/n} + ct^2}\norm{L}_{t_1}\nonumber \\ 
\norm{L_1*L_2}_t&\leq dt^{-1} \norm{L_1}_{t_1} \norm{L_2}_{t_2}.
}

The proof breaks down into three lemmas which respectively provide positive
constants  $a_1, b_1$  and $b_2,c_2$  such
that for  sufficiently small $t $ (depending on $a_i, b_i$), and 
 for every partition
$P$ of $t$
\be
\norm{K^{*P}- K(t)}_t\leq
a_1t^{1+\epsilon};
\ee{norm-k*p}
for every $t_1+t_2=t,$  every pair of partitions $P_1$ of $t_1$ and
$P_2$ of $t_2$, and every $L \in E_{x,y}$
\eqa{norm-k*pl}{
\norm{K^{*P_1}*L}_t &\leq
e^{b_1t_1^{\epsilon/n}+ct^2}\norm{L}_{t_2} \nonumber \\
\norm{L*K^{*P_2}}_t &\leq
e^{b_1t_2^{\epsilon/n}+ ct^2}\norm{L}_{t_1};
}
and  for all partitions $P$ of $t$ and all refinements $Q$ of $P,$
\be
\norm{K^{*Q}- K^{*P}}_t\leq
a_1e^{b_2t ^{\epsilon/n}+ c_2t^2}t\abs{P}^\epsilon.
\ee{norm-refine}
Since $t$ is bounded this implies Eq.~\eqref{eq:k*k-refine}.
\end{proof}

The proofs of Eqs.~\ref{eq:norm-k*p} and~\ref{eq:norm-refine}, provided
in the lemmas below, rely on the
following induction process:  For any partition $P$ of $t$, there
exists a $t_k$ in $P$ such that, writing $P=P_1(t_k)P_2$, with $P_i$ a
partition of $\tau_i$ and $t=\tau_1 + t_k + \tau_2$,  $\tau_i
\leq t/2$. (the largest $k$ so that $\tau_1 < t/2$.)  Assuming
a given property holds for $P_1$ and $P_2$ and proving it therefore holds for $P$
will prove it holds in general, provided it also holds for the empty partition of
$0$ and the trivial partition $(t).$

\begin{lemma} \label{lm:norm-k*p}
Eq.~\eqref{eq:norm-k*p} holds.
\end{lemma}

\begin{proof}
Begin with the special case $P=(t_1,t_2,t_3)$. Several applications of Eq.~\eqref{eq:abc} imply there is an
$a_2$ depending on $a$ and $b$ such that
\be
\norm{K^{*P}- K(t)}_t \leq a_2 t^{1+\epsilon}
\ee{kkk-est}
in this case.

To begin the inductive argument, note Eq.~\eqref{eq:norm-k*p} clearly holds
when $t=0$ or $P=(t)$.
For the inductive step, write
$P=P_1(t_k)P_2$, where $P_i$ is a partition of $\tau_i \leq t/2
$. Using the  decomposition
\[
\abs{ABC-A'BC'} \leq \abs{(A-A')BC'} + \abs{A'B(C-C')} +
\abs{(A-A')B(C-C')},
\]

\begin{align*}
&\norm{K^{*P} - K(t)}_t \\
&\qquad = \norm{K(\tau_1)*K(t_k)*K(\tau_2) - K(t) + K^{*P_1}*
  K(t_k)*K^{*P_2} - K(\tau_1)*K(t_k)*K(\tau_2)}_t \\
&\qquad \leq  \norm{K(\tau_1)*K(t_k)*K(\tau_2) - K(t)}_t +
a_1\tau_1^{1+\epsilon} e^{bt_k^{\epsilon/n}+
  b \tau_2^{\epsilon/n} + 2ct^2}\\
&\qquad \quad   +
a_1 e^{b\tau_1^{\epsilon/n}+
  b t_k^{\epsilon/n} + 2ct^2} \tau_2^{1+\epsilon} +
    d_1 t^{-1}a_1^2 e^{bt_k^{\epsilon/n} + ct^2}\tau_1^{1+\epsilon}\tau_2^{1+\epsilon}
      \\
&\qquad \leq a_2 t^{1+\epsilon} 
+ a_1 e^{2b t^{\epsilon /n } + 2ct^2}   2^{-\epsilon} t^\epsilon \tau_1 + a_1
e^{2b t^{\epsilon /n } +2ct^2}  2^{-\epsilon} t^\epsilon \tau_2  
+ d_1 a_1^2 e^{b t^{\epsilon / n}  +ct^2}  2^{-2 -
  2 \epsilon} t^{1+2\epsilon} \\
&\qquad  \leq a_2 t^{1+\epsilon} +  a_1 2^{-\epsilon} e^{2b t^{\epsilon
    /n } + 2ct^2}    t^{1+\epsilon
}\sbrace{1 + d_1 a_1 2^{-2 -  \epsilon} 
t^{\epsilon}}.
\end{align*}
The second inequality uses $t_k < t$ and $\tau_i < t/2,$ the third is
just factoring plus some easy estimates.  For this to be less than
$a_1t^{1+\epsilon}$ requires that $f(a_1,t)$ be greater than $a_2,$
where $f(a_1,t)= a_1 -a_1 2^{-\epsilon} e^{2b t^{\epsilon
    /n }+2ct^2} \sbrace{1 + d_1 a_1 2^{-2 - 2 \epsilon} 
t^{\epsilon}}.$  When $t=0$ this is true for sufficiently large $a_1,$
so by continuity there is a value of $a_1$ for which it is true for
all $t$ less than some bound. 
\end{proof}


\begin{lemma} \label{lm:norm-k*pl}
Eq.~\eqref{eq:norm-k*pl} holds.
\end{lemma}

\begin{proof}
Using Eqs.~\eqref{eq:norm-k*p} and \eqref{eq:abc}
\begin{align*}
\norm{K^{*P_1}*L}_t &\leq 
\norm{K(t_1)*L}_t + d_1 t^{-1}a_1t_1^{1+\epsilon}\norm{L}_{t_2}\\
&\leq \sbrace{ e^{bt_1^{\epsilon/n} + ct^2} + d_1 a_1
  t_1^\epsilon} \norm{L}_{t_2}\leq e^{b_1 t_1^{\epsilon/n} + ct^2} \norm{L}_{t_2}.
\end{align*}
The second inequality use Eq.~\eqref{eq:abc} and bound $t_1$
by $t,$  and in the
third choose $b_1$ large enough that $e^{b_1 t_1^{\epsilon/n}}$ bounds
$\sbrace{ e^{bt_1^{\epsilon/n}} + d_1 a_1
  t_1^\epsilon} .$
\end{proof}

\begin{lemma} \label{lm:norm-refine}
Eq.~\eqref{eq:norm-refine} holds.
\end{lemma}

\begin{proof}
Write $Q=Q_1Q_2\cdots Q_m$ where $Q_i$ is a partition of $t_i$ in $P.$
Again proceed by induction on $P,$ the empty case being trivial and the
$P=(t)$ case following from Eq.~\eqref{eq:norm-k*p}.   For the induction step, write $P=P_1(t_k)P_2$ with $P_i$ a
partition of $\tau_i \leq t/2,$ and this time apply the decomposition
\begin{align*}
\abs{ABC-A'B'C'}&=\abs{(A-A')BC'} + \abs{ A'B(C-C')} + \abs{ A'(B-B')C'} \\
&+ \abs{
  (A-A')B(C-C')}
\end{align*}
to get
\begin{align*}
&\norm{K^{*Q} - K^{*P}}_t = \norm{K^{*Q_1\cdots Q_{k-1}} * K^{*Q_k} * K^{*Q_{k+1}
      \cdots Q_n} - K^{*P_1}*K(t_k)*K^{*P_2}}_t \\
&\qquad \leq e^{b_1 t_k^{\epsilon/n} + b_1 \tau_2^{\epsilon/n} + b_2
  \tau_1^{\epsilon/n} + c_2\tau_1^2 + 2c t^2}
  a_1 \tau_1 \abs{P_1}^\epsilon
+ e^{b_1 \tau_1^{\epsilon/n} + b_1 t_k^{\epsilon/n} + b_2
  \tau_2^{\epsilon/n}  + c_2 \tau_2^2 + 2ct^2} a_1 \tau_2 \abs{P_2}^\epsilon
\\
&\qquad \qquad     +
e^{b_1\tau_1^{\epsilon/n} + b_1 \tau_2^{\epsilon/n} +2 ct^2 }a_1t_k^{1+\epsilon}  + d_1 a_1^2t^{-1}
e^{b_1
  t_k^{\epsilon/n} + b_2 \tau_1^{\epsilon/n} + b_2 \tau_2^{\epsilon/n}
 + ct^2} \tau_1 \tau_2 \abs{P_1}^{\epsilon}
\abs{P_2}^{\epsilon}  \\
&\qquad \leq a_1e^{2b_1
  t^{\epsilon/n} + 2^{-\epsilon/n} b_2 t^{\epsilon/n} + (2c + 2^{-2} c_2)t^2} \abs{P}^\epsilon \sbrace{\tau_1 + \tau_2 + t_k +
  \frac{d_1 a_1}{4} e^{b_2 t^{\epsilon/n}}t\abs{P}^{\epsilon}} \\
&\qquad \leq a_1e^{(2b_1
  + 2^{-\epsilon/n} b_2 + b_3) t^{\epsilon/n} + (2c + 2^{-2} c_2)t^2}
t \abs{P}^\epsilon \\
&\qquad \leq a_1e^{ 2^{-\epsilon/(2n)} b_2
  t^{\epsilon/n} + c_2 t^2} t\abs{P}^\epsilon.
\end{align*}
The first inequality follows from the inductive hypothesis
Eq.~\eqref{eq:norm-refine}, Eqs.~\eqref{eq:abc}, and 
Eq.~\eqref{eq:norm-k*p}. The second follows from  $\tau_i < t/2$ and
$t_k <  t$ as before, along with $t_k < \abs{P}$, $\abs{P_i} <
\abs{P}$, and $\tau_1\tau_2\leq
t^2/4.$  The third inequality bounds $\abs{P}$ by $t,$ and chooses
$b_3$ so that $e^{b_3t^{\epsilon/n}} \geq 1 + \frac{d_1 a_1}{4} e^{b_2
    t^{\epsilon/n} }t^\epsilon.$  We can clearly choose $b_2,c_2$ big enough to get the fourth
inequality, and  
this last quantity is clearly less than
$a_1 e^{b_2 t^{\epsilon/n}}t\abs{P}^\epsilon$ for
sufficiently small $t.$
\end{proof}


\section{The Large Partition Limit}

\begin{theorem} \label{th:kinf}
For sufficiently small, positive $t$ the pointwise limit
\be K^\infty(t)= \lim_{\abs{P} \to 0} K^{*P}
\ee{kinf-limit}
exists in $E_{x,y}$ and is
smooth in $t.$ In fact, for all partitions and any $0<\epsilon<1/2$ 
\be
\norm{K^\infty(t)- K^{*P}}_t = \OO\parens{t\abs{P}^\epsilon},
\ee{kinf-est}
and, in particular,
\be
\norm{K^\infty(t)- K(t)}_t = \OO\parens{t^{1+\epsilon}}.
\ee{kinf-kest}
Moreover, the limit $K^\infty(t)$ 
is the heat kernel for the Laplace-de Rham operator, satisfying 
\be
\frac{\partial}{\partial t} K^\infty = -\frac{1}{2} \Delta_\text{LdR} K^\infty 
\ee{kinf-heat}
and for every form $f\in \Gamma(\Lambda^* TM)$
\be
\lim_{t \to 0} K^\infty(t)*f = f. 
\ee{kinf-ident}  
\end{theorem}

\begin{proof}
For a fixed $t,$ The expression $H(t) + Dt$ used to define the
$t$-norm in Prop.~\ref{pr:norm} is bounded above and below, so  the
$t$-norm $\norm{\,\cdot\,}_t$ is bounded above and below by a multiple
of the
supremum norm. Therefore, fixing some sequence $P_n$ of refinements
with $\lim_{n \to \infty} \abs{P_n}=0$ the sequence
$K^{*P_n}$ is Cauchy by Eq.~\eqref{eq:k*k-refine} and
thus converges to some $K^\infty(t)$ continuous in $t$ and the
suppressed $x$ and $y.$   In particular
\[\norm{K^\infty(t) - K^{*P_n}}_t= \OO(t)
\abs{P_n}^\epsilon.\]
 If $P$ is any other partition of $t$ then for
each $P_n$ with $\abs{P_n}\leq \abs{P}$ there is a $P'$ which is a common refinement of $P$ and
$P_n,$ and therefore by   Eq.~\eqref{eq:k*k-refine} again all the
quantities $K^{*P'},$ $K^{*P},$
$K^{*P_n},$ and $K^\infty(t)$ differ by
$\OO(t)\abs{P}^\epsilon,$ proving Eqs.~\eqref{eq:kinf-limit},
\eqref{eq:kinf-est}, and \eqref{eq:kinf-kest}.

If $f(x, \psi_x)$ is any smooth form on $M,$ there exists for small
enough $t>0$ a smooth form $f(x,\psi_x;t)$ satisfying the heat
equation 
\[\frac{\partial}{\partial t} f = -\frac{1}{2} \Delta_\text{LdR} f\]
and 
\[ f(x,\psi_x;0)=f(x,\psi_x).\]

By Eqs.~\eqref{eq:op-limit} and~\eqref{eq:kinf-est}
\[K^\infty(t)*f=f(t)\]
proving Eq.~\eqref{eq:kinf-ident}. Thus as a distribution $K^\infty$ is a
solution to Eq.~\eqref{eq:kinf-heat}, and therefore, since the
Laplace-de Rham heat equation is elliptic, by elliptic regularity (\cite{Evans98})
 it is in
fact smooth in $x,$ $y,$ and $t$ and satisfies Eq.~\eqref{eq:kinf-heat}.
\end{proof}

\begin{corollary}   \label{cr:kinf}
In fact $K^\infty$ is defined for all positive $t$ by
\be
K^\infty(t)= \lim_{\abs{P} \to 0} \sbrace{K(t_i)}^{*P},
\ee{kinf-lim2}
and is  the heat kernel
satisfying Eqs.~\eqref{eq:kinf-heat} and~\eqref{eq:kinf-ident}.  The
above limit converges
pointwise and, for fixed $t,$ uniformly.
\end{corollary}
\begin{proof}
Suppose Theorem~\ref{th:kinf} applies for all positive $t$ less than
some $T>0.$ For such $t,$ $K^\infty$ satisifies the
semigroup property $K^\infty(t_1)*K^\infty(t_2)=K^\infty(t)$ for
$t=t_1+t_2$, since the two kernels give the same operator on
forms. Extending the definition of $K^\infty(t)$ for all
positive $t$ to agree with
$(K^\infty)^{*P}$ for $P$ any partition of $t$ fine enough
that $\abs{P}<T$ is therefore well-defined.  This extended $K^\infty(t)$ still satisfies
Eqs.~\eqref{eq:kinf-heat} and~\eqref{eq:kinf-ident}.

To see that this extended $K^\infty(t)$ satisfies
Eq.~\eqref{eq:kinf-lim2}, consider any partition $Q$ with
$\abs{Q}<T/2,$ so that
it can be written as a refinement of a partition $P$ where each each
$t_i$ satisfies $T/2 \leq t_i <T.$  That is, $Q=Q_1\cdots Q_m$, where
$Q_i$ is a partition of $t_i$. Eq.~\eqref{eq:kinf-est} and
$\abs{Q_i} \leq \abs{Q}$, guarantee
\[\norm{K^{*Q_i} - K^\infty(t_i)}_{t_i} = \OO\parens{t_i}
\abs{Q}^\epsilon.\]
Since $t_i$ is bounded above and below  $\norm{\,\cdot
  \,}_{t_i}$ is commensurate with  the sup norm,
so there is a $b>0$ such that 
\[\norm{K^{*Q_i} - K^\infty(t_i)}_\infty \leq bt_i\abs{Q}^\epsilon.\]
Also, because $K^\infty$ is the kernel of an elliptic operator there
is a $c>0$ such that $\norm{K^\infty(t)*f}_\infty \leq e^{ct}
\norm{f}_\infty.$   Induction on $m$ will prove
$\norm{K^{*Q}-K^\infty(t)}_\infty \leq e^{at}\abs{Q}^\epsilon.$
For the induction step,  apply the decomposition
\[\abs{AB-A'B'}\leq \abs{(A-A')B'} + \abs{A'(B-B')} +
\abs{(A-A')(B-B')}\]
to get 
\begin{align*}
\norm{K^{*Q} - K^\infty(t)}_\infty & =
\norm{K^{*Q1\cdots Q_{m-1}} * K^{*Q_m} -
  K^\infty(t-t_m) * K^\infty(t_m)}_\infty\\
& \leq \sbrace{ e^{a(t-t_m)+ ct_m} + e^{c(t-t_m)} bt_m + \vol(M)e^{a(t-t_m)} b t_m\abs{Q}^\epsilon}\abs{Q}^\epsilon\\
& \leq e^{a(t-t_m)}\sbrace{e^{ ct_m} +  b_1 t_m}\abs{Q}^\epsilon\\
& \leq e^{at}\abs{Q}^\epsilon
\end{align*}
assuming in the second inequality that $a\geq c$ and choosing $b_1\geq
b+ \vol(M)bT^\epsilon$ 
and in the third assuming that
$a \geq c+b_1.$  Eq.~\eqref{eq:kinf-lim2} follows.
\end{proof}
\begin{remark}
In \cite{FS08} the authors prove a weaker result for a kernel $K^{mq}$
similar
to $K$. There the result was that the limit taken along {\em some sequence} of
refinements of $P$ exists and agrees pointwise with
the heat kernel. The stronger result depends on modifying the
approximating kernel to agree with the heat kernel (as a distribution)
to higher order in $t$. Both approximate kernels are consistent with
time-slicing the imaginary-time path integral, as they amount to different
interpolations between the fixed end-values.
\end{remark}
\begin{remark}  The choice of $K$ as the approximate kernel, and the
  basis for a discrete
 approximation to the path integral, is, as noted in the introduction,
 subject to some 
 ambiguity in how the path integral specifies a discrete
 approximation.  It is natural to ask how the results of
 Theorem~\ref{th:kinf} and Corollary~\ref{cr:kinf} depend on the
 exact choice of $K.$  Most straightforwardly, adding terms of the
 form $\OO\parens{t^2}H$ or $\OO\parens{\abs{\vec{y}_x}^3}H,$ and using  the standard norm in place of $\abs{\,\cdot\,}_t$, it is
 easy to follow along with the argument and see it all goes through unchanged. This
 would not be sufficient to prove the Gauss-Bonnet-Chern Theorem, but
 in most other respects seems as powerful a result.  More subtly,
 adding to $K$ a term of the form $H \sbrace{\parens{\vec{y}_x, A \vec{y}_x} - t \tr
 A},$ where $A$ is section of $\End\parens{T_xM},$ does not change
the limit, although the above argument in itself would not suffice and
the convergence is slower.  In particular Eq.~\eqref{eq:kinf-kest}
would no longer hold (even with a modified norm).  Thus while $K=H
e^{t\scalar/6}$, which would appear to be
the simplest choice, will  converge in the large partition limit to the
heat kernel for the Laplace-Beltrami operator, only $K=He^{t\scalar/12 +
 t\Ricci\parens{\vec{y}_x, \vec{y}_x}/12}$ offers an analogue of
Eq.~\eqref{eq:heat-est} and hence
Eqs.~\eqref{eq:kinf-est} and~\eqref{eq:kinf-kest}.
\end{remark}

\section{The Gauss-Bonnet-Chern Theorem} \label{sc:gbc}

\begin{theorem} \label{th:gbc}
\be
\sbrace{\lim_{t\to 0}  K^\infty(x,x,\psi_x,\psi_x;t)}_\mathrm{top}
= \parens{ 2\pi}^{-n/2} \Pfaff(R)
\ee{gbc}

\end{theorem}
\begin{remark}
McKean and Singer \cite{MS67} prove the integral of the matrix
supertrace of the heat kernel, i.e. $\int \bint
K^\infty(x,x,\psi_x,\psi_x;t) d\psi_x\, dx,$ 
is equal to the signed sum of the Betti numbers, independent of $t.$ Thus,
the theorem implies the Gauss-Bonnet-Chern
theorem, thereby completing a rigorous version of the path-integral
proof of the latter. In fact, it implies a local version of the
Gauss-Bonnet-Chern theorem: 
\[
\lim_{t\to 0} \str K^\infty(x,x;t) = \parens{ 2\pi}^{-n/2} 
\Pfaff(R)(x).
\]
\end{remark}
\begin{proof}
Theorem~\ref{th:kinf}, in particular
Eq.~\eqref{eq:kinf-kest}, estimates the heat kernel $K^\infty$ 
as the approximation $K$ plus an error term whose $t$-norm is
$\OO(t^{1+\epsilon})$. That is, $K^\infty(t)- K(t) = F_1H + F_2t$, where $\abs{F_i}_t=\OO\parens{t^{1+\epsilon}},$ so
the supertrace of the error term is bounded by
\[\str{F_1} \parens{2\pi t}^{-n/2} + \str{F_2}t.\]
Lemma~\ref{lm:str-est} implies
$\str{F_i}=\OO\parens{t^{n/2+2\epsilon/n}} $ so the supertrace of the
error term is bounded by $\OO\parens{t^{2\epsilon/n}}$ and therefore
goes to zero with $t$.  On the other hand by Eq.~\eqref{eq:K-def}
\begin{align*}
\lim_{t\to 0} \bint K(x,x,\psi,\psi;t) &= \lim_{t\to 0} \parens{2 \pi t}^{-n/2}
\bint \exp\sbrace{
-\frac{t\scalar }{6} 
+ 
\frac{t}{4} \parens{\rho, R\sbrace{\psi, \psi} \rho} 
}d \rho\\
&= \parens{2 \pi}^{-n/2} \bint e^{\frac{1}{4}\parens{\rho, R\sbrace{\psi,
      \psi} \rho} } d\rho =\parens{2 \pi}^{-n/2} \Pfaff\parens{\frac{1}{2}R\sbrace{\psi,
    \psi}} \\
& = \parens{2 \pi}^{-n/2} \Pfaff\parens{R}.
\end{align*}
\end{proof}
\section{Some conclusions} \label{sc:concl}
The heuristic path integral argument leading to the Gauss-Bonnet-Chern
and other index theorems has two key premises: first, the path
integral for  supersymmetric quantum mechanics, taken over free loops,
gives the supertrace of the Laplace-de Rahm heat kernel; second,
steepest descent applies to the path integral to give the Pfaffian of
the curvature as its small-$t$ limit. Proposition~\ref{pr:heat-est}
and Eq.~\ref{eq:kinf-kest} of Theorem~\ref{th:kinf}
give the precise sense in which the kernel $K$ approximates the heat
kernel for small $t$. 
The product $K^{*P}$ thus
represents the approximate imaginary-time path integral for paths with
fixed endpoints, in accord 
with Feynman's time-slicing approach to combining a choice of
partition $P$ with the small-$t$
approximation. Corollary~\ref{cr:kinf} to Theorem~\ref{th:kinf}
then proves, that, with this precise definition, the path integral with fixed endpoints agrees with the
heat kernel. This immediately leads to the first premise of the heuristic
argument. For the second, Eq.~\ref{eq:kinf-kest} proves the
steepest-descent approximation applies to this definition of the path
integral, and Theorem~\ref{th:gbc} shows this approximation agrees
with the Pfaffian in the small-$t$ limit. 

In short, $\lim_{\abs{P} \to
  0}K^{*P}$ provides a rigorous definition for the path integral in
which the path integral argument for GBC carries over directly. As
such, it may serve as a template for rigorous definition path
integrals for such immediate generalizations as $N=1/2$ supersymmetric
quantum mechanics, which the authors believe requires only minor
changes to the definition of $K$ and should give a path integral proof
of the Atiyah-Singer index theorem. More generally, it is not
unreasonable to hope this approach will generalize to a wide range of
cohomological field theories, because each of these depend, at least
heuristically, on the same localization to the steepest-descent (or
stationary-phase) approximation.

Looked at another way, the heuristic arguments rely on an unstated
assumption about the time-slicing prescription for the path integral;
namely, the order of taking the small-$t$ and 
small-$\abs{P}$ limits does not matter. The application of
Lemma~\ref{lm:str-est} in the proof of Theorem~\ref{th:gbc}
shows, in conjunction with the Cauchy property of
Proposition~\ref{pr:cauchy},  that, even for the path integral
corresponding to based loops, the small-$t$ limit of $K^\infty$  
agrees with the small-$\abs{P}$ limit of $K^{*P}(0)$. 

In their seminal work on the heat kernel McKean and Singer \cite{MS67}
referred to ``fantastic cancellations''
which must occur amongst lower-order terms in the asymptotic
expansion for the heat kernel so that its supertrace
agrees with that of the Pfaffian of the curvature, in keeping with the
Gauss-Bonnet-Chern theorem. In the present construction of the path integral, according to
Eq.~\ref{eq:kinf-kest}, the asymptotic expansion for the 
heat kernel begins with terms which on the diagonal are
$\OO(t^{-n/2})$, yet the supertrace is $\OO(1)$. This happens
because, first, the approximating kernel $K$ already has
this property, and second, the $t$-norm was devised so that the error
terms, which are $\OO_t(t^{1+\epsilon})$, have supertrace in
$\OO(t^{\epsilon/n})$. 
\begin{acknowledgements} 
 The first author would like to thank his home
  institution for granting him a sabbatical leave, and the Department of
  Mathematics at M.I.T. for kindly hosting him as a visitor.  The
  second author would like to thank Lisa Sawin for project management advice.
\end{acknowledgements}

\def\cprime{$'$}

\end{document}